\documentclass[12pt]{article}

\usepackage{graphicx}
\usepackage{amssymb}
\usepackage{fancyhdr}
\usepackage{xcolor}
\usepackage{url}
\usepackage{mathtools}
\usepackage{amsmath}
\usepackage[utf8]{inputenc}
\usepackage[english]{babel}
\usepackage{breakcites}
\usepackage{tabu}
\usepackage{mathrsfs}  
\usepackage{booktabs}
\usepackage{amsthm}
\usepackage{hyperref}
\usepackage{rotating}
\usepackage{graphicx}
\usepackage{lscape}
\usepackage{setspace}

\definecolor{cornellred}{rgb}{0.7, 0.11, 0.11}
\hypersetup{colorlinks,linkcolor={red},citecolor={cornellred},urlcolor={red}}

\newtheoremstyle{break}
  {\topsep}{\topsep}%
  {\itshape}{}%
  {\bfseries}{}%
  {\newline}{}%
\theoremstyle{break}
\newtheorem{definition}{Definition}

\newtheorem{theorem}[definition]{Theorem}
\newtheorem{lemma}[definition]{Lemma}

\newtheorem{assumption}{Assumption}

\DeclareMathOperator{\spargel}{sp}

\DeclareMathOperator{\diag}{diag}
\DeclareMathOperator{\rank}{rk}
\DeclareMathOperator{\E}{\mathbb E}
\DeclareMathOperator{\proj}{proj}

\DeclareMathOperator{\vech}{vech}
\DeclareMathOperator{\vect}{vec}

\makeatletter
\renewenvironment{proof}[1][\proofname]{%
  \par\pushQED{\qed}\normalfont%
  \topsep6\p@\@plus6\p@\relax
  \trivlist\item[\hskip\labelsep\bfseries#1\@addpunct{.}]%
  \ignorespaces
}{%
  \popQED\endtrivlist\@endpefalse
}
\makeatother

\newcommand\norm[1]{\left\lVert#1\right\rVert}

\usepackage[round]{natbib}

\makeatletter
\newcommand{\bianca}{\renewcommand\NAT@open{[}\renewcommand\NAT@close{]}}
\makeatother

\usepackage{stackengine}

\begin{document}

\title{Retrieval from Mixed Sampling Frequency:\\ Generic Identifiability in the Unit Root VAR 
} 

\author{Philipp Gersing\footnote{Department of Statistics, Vienna University of Technology, Institute for Statistics and Mathematics, Vienna University of Economics and Business.
Email: philipp.gersing@wu.ac.at.
}, 
Leopold S\"ogner\footnote{Department of Economics and Finance, Institute for Advanced Studies,
Josefst\"{a}dter Stra\ss{}e 39,
1080 Vienna, Austria. 
Leopold S\"ogner has a further affiliation with the Vienna Graduate School of Finance (VGSF). 
Email: soegner@ihs.ac.at.},
Manfred Deistler\footnote{Department of Statistics, Vienna University of Technology, Institute for Statistics and Mathematics, Vienna University of Economics and Business;
Email: manfred.deistler@tuwien.ac.at
}}

	\vspace*{-150pt}
	{\let\newpage\relax\maketitle}
	\maketitle
	
%


%
\vspace*{-5pt}
 %
\begin{abstract}
The ``\textit{RE}trieval from \textit{MI}xed Frequency \textit{S}ampling'' (REMIS) approach based on blocking developed in \cite{AndersonEtAl2016a} is concerned with retrieving an underlying high frequency model from mixed frequency observations. 

In this paper we investigate parameter-identifiability in the \citet{Johansen1995} vector error correction model for mixed frequency data.
We prove that from the second moments of the blocked process after taking differences at lag $N$ ($N$ is the slow sampling rate), the parameters of the high frequency system are generically identified. We treat the stock and the flow case as well as deterministic terms.%
\end{abstract}
{\em Keywords:}{\em Keywords:} Mixed Frequency, REMIS, VAR, Cointegration, Vector Error Correction Model, Identifiability

\noindent {\em MSC:} 62M10, 62P20

%
%
\newpage
\section{Introduction} 
\label{sect:intro}
%
%
%

Econometric analysis is often encountered with multivariate time series data sampled at mixed frequencies. Examples for treating this are \cite{Zadrozny1988}, \citet{GhyselsEtAl2007}[MIDAS-regression],
\cite{anderson2012identifiability}, \cite{Schorfheide2015Real-Time},  \citet{ghysels2016macroeconomics}, \cite{AndersonEtAl2016a} and \cite{chambers2020frequency}. Identifiability is a prerequisite for consistent estimation \citep[see, e.g.,][]{DeistlerSeifert1978, PoetscherPrucha1997} and often is needed for economic interpretation of effects related to particular model parameters. This article investigates {\em identifiability} of the model parameters in a \citet{Johansen1995} vector error correction model.\\
The general question is whether the internal characteristics, i.e. the model parameters $\theta$, can be retrieved from the external characteristics -- in our case observable second moments. Identifiability means that the mapping from the parameters to these second moments is injective. Often injectivity of this mapping can only be achieved for a certain subset of the parameterspace. Here, we prove that identifiability can be obtained for a generic subset of the parameterspace \citep[see][]{AndersonEtAl2016a}.\\
%
%
%
%
%
As opposed to MIDAS-regression, where the observations at high frequency are considered as additional information, we consider mixed frequency as either a ``missing-values'' or a ``dis-aggregation''-problem, by which we mean the following: 
%
We commence from an underlying \textit{high frequency system} (e.g., a VECM) parameterised by $\theta$ for a multivariate process 
\begin{align*}
(y_t)_{t \in \mathbb Z} = \left( \begin{pmatrix}
 y_t^f \\
 y_t^s
\end{pmatrix} \right)_{t \in \mathbb Z}  \ ,   
\end{align*}
with dimensions $n$, $n_f$ and $n_s$ for $y_t$, $y_t^f$ and $y_t^s$ respectively. Our aim is to identify and estimate the high frequency system from the observed (mixed frequency) data. 
The \textit{observational scheme} is as follows: 
While the fast variables $y_t^f$ are observed at $t \in \mathbb{Z}$, for the slow variables 
$y_t^s$ we consider: \\
\noindent 1. \textit{Stock-Case:}  $y_t^s$ is observed only at $t \in N \mathbb Z$ for some sampling rate $N \geq 2$, hence we have a \textit{missing-value problem}. \\
\noindent
2. \textit{Affine aggregation:} we observe an affine transformation 
\begin{align}
w_t := c_w + c_0 y_{t}^s + \cdots + c_{p_c} y_{t-p_c}^s \ , \label{eq: def w_t} 
\end{align}
where $c_i$ are known constant matrices for $i \geq 0$, $c_w$ is a known vector and $w_t$ is observed at $t \in N \mathbb Z$. A special case of affine aggregation are flow variables: For example suppose $y_t^s = GDP_t$, the monthly gross domestic product of a country. The quarterly GDP, $w_t$, is the sum of three monthly GDPs. We call $y_t^s$ \textit{latent} whenever it is not directly observed. Hence, our aim is to retrieve the underlying high frequency parameters $\theta$ from data observed according to the observational schemes described above.%
\footnote{In this example we assume that the variable considered, 
$y_t$, is integrated of order one. If by contrast $(\log y_t)$ is integrated of order one, the affine approximation of \citet{Aadland2000} in combination with the methodology developed in this article can be applied.}
\\  
With the procedure described above, we are able to model all kinds of linear dynamic relationships between latent and observed variables, whereas the MIDAS \citep[see, e.g., ][]{ghysels2016macroeconomics} approach only covers relationships between observed variables. After identifying the parameters one may interpolate missing values or dis-aggregate observations in a model based way by using the retrieved parameters of the underlying high frequency system.\\ 
%
%
%
%
%
%
%
%
Estimation of continuous time models from mixed frequency data are investigated in \citet{chambers2003asymptotic, chambers2016estimation, chambers2020frequency}. In particular, \citet{chambers2003asymptotic, chambers2020frequency} consider co-integrating regressions and show that the scaled estimators proposed, converge in distribution to functionals of Brownian motion and to stochastic integrals. Hence, the estimators are (weakly) consistent. Then, by
\citet{GABRIELSEN1978} -- and for the case of strong consistency by
\citet{DeistlerSeifert1978} -- the model parameters are identified.\\  
%
%
%
%
For the stable vector auto-regessive model \citet{anderson2012identifiability} and \citet{AndersonEtAl2016a} either used the {\em blocking approach} \citep[see also][]{Filler2010, ghysels2016macroeconomics} or the {\em extended Yule-Walker equations} \citep[see][]{ChenZadrozny1998, AndersonEtAl2016a} to show g-identifiability. For the same model class \citet{GersingAndDeistler2021} present an alternative proof for identifiability using the so-called canonical projection form. This idea is also applied in this paper. On the other hand, \citet{DeistlerEtAl2017} show that the parameters need not be identified in the auto-regressive-moving average (VARMA) case, if the order of the MA polynomial exceeds the order of the AR polynomial.\\ 
This article is organised as follows: Section~\ref{sect: notation and model class} starts with the vector error correction model developed in \citet{Johansen1995} as the underlying high frequency model. 
%
%
%
%
In Section~\ref{subsect: mixed frequency data and generic ident} we describe the observational schemes considered in detail. In particular, we introduce a stationary blocked process containing all observed variables. Section~\ref{sec: topological props} introduces conditions, which are later shown to be sufficient for identifiability. We prove that these conditions hold generically in the underlying high frequency parameterspace. Section~\ref{sec: super idea} extends the REMIS approach to the non-stationary case: Here, we use the result from \cite{chambers2020frequency} that the cointegrating vectors can be identified from mixed frequency data. First, we derive a state-space representation of the blocked process that we call Canonical Projection Form (CPF). In this representation, the system matrices are simple transformations of the parameters of the underlying high-frequency model. After that we start from the unique factor of the spectrum of the blocked process \citep[see, e.g.,][chapter 6.2 and 7.3]{ScherrerDeistler2018book} to get an arbitrary minimal realisation for this factor and relate this to the CPF. From there we can retrieve the parameters of the underlying high frequency system using the structural properties of the CPF. Section~\ref{sect:deter} adds deterministic terms. Finally, Section~\ref{sect:conc} concludes.

\section{Notation and Model Class}\label{sect: notation and model class}
%
%
%
%
%
%
\subsection{Representations and Parameterspace of the Underlying High Frequency System}
\label{sect:model}
In the first step, we introduce the class of underlying high frequency systems: We commence from a process which is integrated of order one and allows for cointegration. Suppose $(y_t)_{t \in \mathbb Z}$ is $n \times 1$ and a solution on $\mathbb Z$ of the vector error correction system: 
\begin{align}
	\Delta {y}_{t}
	=  \Pi {y}_{t-1}   
	+ \sum_{j=1}^{p-1} {\Phi}_{j}  \Delta 
	{y}_{t-j}
	+ { \nu}_{t},  \quad \ \ \nu_t \sim WN \left( \Sigma_{\nu} \right) \ , 	\label{eq: VECM hf}
\end{align}  
%
where $(\nu_t)_{t \in \mathbb Z}$ is white noise and $\Pi$ is of rank $r > 0$ in the case of cointegrating relationships, but we also allow the case $r = 0$. Such solutions always exist and can be constructed as described in detail in \cite{bauerwagner2012}. We obtain a unique factorisation of $\Pi = \alpha \beta'$ with $\alpha, \beta \in \mathbb{R}^{n \times r}$ applying the singular value decomposition to $\Pi$ in the following way: %
\begin{align*}
\Pi &=  \underbrace{U}_\text{$n \times n$}
\underbrace{\diag(d_1, ..., d_r, 0, ..., 0)}_\text{$D$}
\underbrace{V'}_\text{$n \times n$}
 = \underbrace{U_1}_\text{$n \times r$} \underbrace{\diag(d_1, ..., d_r)}_\text{$\tilde D$} \underbrace{V_1'}_\text{$r \times n$} \\
 &= U_1 \tilde D V_1'  = \underbrace{U_1  Q^{-1}}_\text{$\alpha$} \underbrace{Q \tilde D V_1'}_\text{$\beta'$},  
\end{align*} 
where $Q$ is a non-singular matrix of elementary row operations that transforms $\tilde D V_1'$ into its reduced echelon form, such that $Q \tilde D V_1' =  \begin{pmatrix} I_r & \beta_{n-r}' \end{pmatrix}$. We stack the parameters $\alpha, \beta, \Phi_1, ..., \Phi_{p-1}$ to a vector $\theta_{VECM} \in \mathbb{R}^{d}$, where $d = nr + (n - r)r +(p-1)n^2$.\\ 
We also have a VAR$(p)$ representation for $(y_t)$ of the form, 
\begin{align}
y_t = \mathcal{A}_1 y_{t-1} + \cdots + \mathcal{A}_p y_{t-p} + \nu_t \ .  \label{eq: hf AR rep} 
\end{align}
Throughout this article, we assume that $r$ and $p$ are known a priori. 
We obtain the representation in (\ref{eq: hf AR rep}) by the mapping $\psi$: 
\begin{align*}
\psi: \ \theta_{VECM}  & \mapsto \theta_{AR}, \quad  \mbox{defined as}\\[0.5em]
\mathcal{A}_1 &= I_n + \Pi + \Phi_1, \quad 
\mathcal{A}_j = \Phi_j - \Phi_{j-1} \ \mbox{  for  } 1 < j < p, \quad 
\mathcal{A}_p = -\Phi_{p-1},  	
\end{align*}
with $\theta_{AR} = \vect \begin{pmatrix} \mathcal A_1 & \cdots & \mathcal A_p\end{pmatrix}$. On the other hand for a $\theta_{AR}$ which has a corresponding VECM representation, we compute $\theta_{VECM}$ as follows: 
\begin{align*}
\psi^{-1}&: \theta_{AR} \mapsto \theta_{VECM} \\[0.5em]
\Pi &= - I_n + \sum_{j = 1}^p	\mathcal{A}_j , \
\Phi_1 = - I_n + \mathcal{A}_1  + \Pi, \
\Phi_2 = \Phi_1 + \mathcal{A}_2, \ \cdots, \ 
\Phi_{p-1} = - \mathcal{A}_p.
\end{align*}
Now, define the polynomial matrix $a(z) = I_n - \mathcal{A}_1 z - \cdots - \mathcal{A}_p z^p$ where $z$ is a complex variable or the lag operator on $\mathbb Z$ depending on the context. For $\check{c} = \begin{pmatrix}
	I_r  \\ 0 
\end{pmatrix} \in \mathbb{R}^{n \times r}$ and 
$\check{c}_\bot = \begin{pmatrix}
 	0 \\ I_{n-r}
 \end{pmatrix} \in  \mathbb{R}^{n \times (n-r)}
$, $\beta_\bot := \big(I_n - \check{c} (\beta' \check{c})^{-1} \beta' \big) \check{c}_\bot$, and 
$\alpha_\bot$ defined analogously to $\beta_\bot$. We impose the following assumptions \citep[][chapter 4]{Johansen1995}:
\begin{assumption}[Cointegrated VAR-System]\label{ass: Theta}\ \\[- 2.2em] 
\begin{itemize}
\item[\textbf{(C1)}] $\rank \alpha \beta ' = r < n$. \label{as: C1}
\item[\textbf{(C2)}] $\det (\alpha_\bot ' (I_n - \sum_{j = 1}^{p-1} \Phi_j) \beta_\bot) \neq 0$. \label{as: C2} 
\item[\textbf{(C3)}] $\det a(z) = 0 \Rightarrow z = 1$ or $|z| > 1$. \label{as: C3} 
\item[\textbf{(C4)}] $\Sigma_\nu = \E \nu_t \nu_t' > 0$. \label{as: C4}
\end{itemize}
\end{assumption}
We define the parameterspace as follows:\footnote{We write $\mathbb R^d \Big|_{C1, C2}$ to denote the set of real vectors in $\mathbb R^d$ for which $C1$ and $C2$ hold.} 
\begin{align*}
 \Theta_{VECM, 1} &:=  \psi^{-1}\left(\psi \left(\mathbb R^d \Big|_{C1, C2}\right) \Bigg|_{C3} \right), \qquad 
 \Theta_{1}:= \psi\left(\Theta_{VECM, 1}\right) \\[.5em]
 &\mbox{with } \quad \Theta_{1} \overset{\psi}{\leftrightarrow} \Theta_{VECM, 1} 
\end{align*}
Note that under these assumptions $\psi$ is a homeomorphism.  
The set of $\vech{\Sigma_\nu}$ with $\Sigma_\nu \in \mathbb{R}^{n \times n}$, $\Sigma_\nu = \Sigma_\nu'$ and $\Sigma_\nu > 0$ (condition (C4) in Assumption \ref{ass: Theta}) is denoted by $\Theta_2$. The overall parameterspace for the VAR$(p)$ representation is
\begin{align*}
 \Theta = \Theta_1 \times \Theta_2.   
\end{align*}
We will also need the state-space representation of $(y_t)_{t \in \mathbb Z}$, which follows from (\ref{eq: hf AR rep}): 
\begin{align}
  \underbrace{\begin{pmatrix}
      y_t \\
      \vdots \\
      y_{t-p+1}
  \end{pmatrix}}_\text{$X_{t+1}$} &= \underbrace{\begin{pmatrix}
    	\mathcal A_1 &  \mathcal  A_2 & \cdots & \mathcal A_{p}   \\
	I_n &  & & 0 \\
	& \ddots & &  \vdots \\
	& &  I_n & 0  
\end{pmatrix} }_\text{$\mathcal  A$}
\underbrace{\begin{pmatrix}
    y_{t-1} \\
    \vdots  \\
    y_{t-p}
\end{pmatrix}}_\text{$X_{t}$} + 
\underbrace{\begin{pmatrix}
    I_n \\
    0 \\
    \vdots \\
    0
\end{pmatrix}}_\text{$\mathcal B$}
\nu_t \label{eq: hf ss trans}\\
y_t &= \underbrace{\begin{pmatrix}
   \mathcal  A_1 & \cdots & \mathcal  A_p
\end{pmatrix}}_\text{$\mathcal  C$} X_{t} + \nu_t. \label{eq: hf ss obs}
\end{align}  
Note that (\ref{eq: hf ss trans}), (\ref{eq: hf ss obs}) is always controllable as $\Sigma_\nu$ and therefore $\Gamma (t) := \E \big(X_{t+1} X_{t+1}'\big) $ are of full rank. The system (\ref{eq: hf ss trans}), (\ref{eq: hf ss obs}) is also observable whenever $\mathcal A_p$ is of full rank. This follows since $\mathcal A_p$ is nonsingular (and therefore $\mathcal A$ is non-singular) from the BPH-test (see \cite{Kailath1980} 2.4.3). Hence under Assumption \ref{ass: Theta} and if $\mathcal A_p$ is nonsingular the system (\ref{eq: hf ss trans}), (\ref{eq: hf ss obs}) is minimal. For details on controllability and observability see e.g. \cite{ScherrerDeistler2018book}, chapter 7 or \cite{HannanAndDeistler2012}, chapter 2.   
\subsection{Mixed Frequency Data: Stock and Flow Variables}\label{subsect: mixed frequency data and generic ident}
A main challenge of the identifiability proof in the integrated case -- as opposed to the stationary case \citep{AndersonEtAl2016a} -- is that the second moments of an integrated process
{(that is, 
$\mathbb{E} y_s y_t$, $s,t \in \mathbb{Z}$)}
are time dependent and cannot be estimated directly. Instead, for the sake of practical relevance of identifiability considerations, we identify from observable second moments of stationary transformations of the level process 
{(that is, 
$\left( y_t \right)_{t \in \mathbb{Z}}$).} \\ 
Suppose for the moment, that the matrix of cointegration vectors $\beta$ is known. Our proof commences from what we call the ``blocked  process'', where we distinguish between the Stock- and the Flow-case:\\
\textbf{1. Stock Variables:} In this case for $t \in N \mathbb{Z}$, we get the co-stationary vector $\tilde{y}_t$ of ``observed'' random variables. We will use $\tilde{n} := r + n + (N-1) n_f$ for the dimension of $\tilde y_t$ henceforth. Let $u_t^\mathcal{S} :=\beta^{\prime} y_t$, $\Delta_N y_t := y_t -y_{t-N} = \sum_{j=0}^{N-1} \Delta y_{t-j}$, {and }\\
\begin{align}
\tilde{y}_{t} & = 
\begin{pmatrix}%
\beta^{\prime} y_t \\
y_t - y_{t-N} \\
y^f_{t-1} - y^f_{t-N-1} \\
\vdots \\
y^f_{t-N+1} - y^f_{t-2N+1} 	
\end{pmatrix} 
=
\begin{pmatrix}%
u_t^\mathcal{S} \\
\Delta_N y_t  \\
\Delta_N y^f_{t-1}  \\
\vdots \\
\Delta_N y^f_{t-N+1} 
\end{pmatrix} 
\ . \label{eq:tildeyt_FDstock}
\end{align}  
The blocked process $(\tilde y_t)$ is similar to the blocked process in \cite{AndersonEtAl2016a} with the distinction that we added the variable $\beta'y_t = u_t^{\mathcal S}$ and take differences at lag $N$. Admittedly, the true $\beta$ is in fact not observed, however since $\beta$ can be estimated consistently, for the purpose of the analysis of identifiability we can assume $\beta'y_t$ to be observed.\\ 
%
%
%
%
%
%
%
%
%
\textbf{2. Flow Variables:}
In a similar way, we may consider the case where all slow variables are flow variables, in which case we are able to observe the temporal aggregate $w_t := \sum_{j=0}^{N-1} y_{t-j}^s$ at $t \in N \mathbb{Z}$. So 
\begin{align*}
   \Delta_N^\Sigma y_t := 
   \sum_{j = 0}^{N-1} y_{t-j} - \sum_{j = 0}^{N-1} y_{t-N-j} =
   \Delta_N \sum_{j = 0}^{N-1} \begin{pmatrix}
    y_t^f \\
    y_t^s
    \end{pmatrix}.
\end{align*}
%
%
%
%
If all slow variables are flow variables, we can observe $\sum_{j=0}^{N-1} y_{t-j} 
= \left( w_t^{\prime} , \sum_{j=0}^{N-1} y_{t-j}^{f \prime} \right)^{\prime}
$, $t \in N \mathbb{Z}$. Since $\beta^{\prime} y_t$ is stationary, we have that
$\left( \beta^{\prime} y_t \right)_{t \in N \mathbb{Z}}$ and $u_t^\mathcal{F} := \beta^{\prime} \sum_{j=0}^{N-1} y_{t-j} \in \mathbb{R}^r$ are integrated of order zero. For the flow case we define the co-stationary vector process
%
%
%
\begin{align}
\tilde{y}_{t}  = 
\begin{pmatrix}%
u_t^\mathcal{F} \\
\Delta_N^\Sigma y_t  \\
\Delta_N y^f_{t-1}  \\
\vdots \\
\Delta_N y^f_{t-N+1} 
\end{pmatrix} 
\ . \label{eq:tildeyt_FD_flow}
\end{align}  
We call the autocovariance function of the (stationary) blocked process
\begin{align}
 \tilde \gamma: h \mapsto  \E \tilde y_{t+h} \tilde y_t'  \ ,\qquad \   \mbox{where  } h \in N \mathbb Z  \ , \label{eq: def tilde gamma} 
\end{align}
\textit{observed second moments}, which can be consistently estimated from the data (if $\beta$ is known) under standard assumptions.%
%
%
\\  
The motivation to consider this blocked process for identifiability is the following:\\
1. We take differences at lag $N$ (as opposed to lag one) because these differences can be directly computed from the mixed frequency data and are stationary.\\
2. Note that the set of observable autocovariances given mixed frequency data is 
\begin{align*}
\gamma_{\Delta_N y}^{ff}(h) &:= \E \Delta_N y_{t+h}^f \Delta_N {y_t^f}'   \qquad h \in \mathbb Z \\
\gamma_{\Delta_N y}^{fs}(h) &:= \E \Delta_N y_{t+h}^f \Delta_N {y_t^s}'   \qquad h \in \mathbb Z \\
\gamma_{\Delta_N y}^{ss}(h) &:= \E \Delta_N y_{t+h}^s \Delta_N {y_t^s}'   \qquad h \in N \mathbb Z \\
\gamma_{\beta}^\cdot (h) &:= \E u_{t+h}^\cdot {u_t^\cdot}'   \qquad \qquad \quad h \in N \mathbb Z \ ,
\end{align*} 
where the superscript ``$\cdot$'' is shorthand for $\mathcal{S}$ or
$\mathcal{F}$.
Note that these are exactly the second moments of the autocovariance function $\tilde \gamma$ of the blocked process defined in equations 
(\ref{eq:tildeyt_FDstock}) for the stock case. In an obvious way this is treated accordingly in the flow case (\ref{eq:tildeyt_FD_flow}). So the blocked process ``contains the whole second moment information available'' from which we can identify. %
The same idea is also applied for the stationary case in \cite{AndersonEtAl2016a}.
\\  
3. Our interest in the particular blocked process (\ref{eq:tildeyt_FDstock}), (\ref{eq:tildeyt_FD_flow}) having $u_t^\cdot$ in the first coordinates, originates in the fact that we can obtain a minimal representation for this process (see Section~\ref{sec: super idea}), where the parameters are fairly simple functions of the parameters of the underlying high frequency system. This will finally help us to retrieve the high frequency model parameters.\\
%
%
%
%
%
%
%
%
%
Next, we define the concept of generic identifiability. Here, identifiability is concerned with the problem whether the parameters of the underlying high frequency system (\ref{eq: hf ss trans}), (\ref{eq: hf ss obs}) or (\ref{eq: VECM hf}) are uniquely determined from the observable second moments (defined below in this section). To be more precise, a subset $\Theta_I \subset \Theta$ is called identifiable, {if the mapping attaching the observable second moments to the parameters $\theta \in \Theta_I$ is injective. 
}
In our setting identifiability for the whole set $\Theta$ cannot be obtained.
%
%
%
%
To see this, we consider a simple example where $p=1$, $r=1$, and $n=2$, the first coordinate of $y_t$ is a fast variable, denoted $y_{t}^{f}$, while the second coordinate, $y_t^s$, is a slow stock variable. We assume that the cointegrating vector $\beta= \left(1,\beta_{s} \right)$ is known. Recall that the observed second moments are as described in equations (\ref{eq:tildeyt_FDstock}) and (\ref{eq: def tilde gamma}). Let $\sigma_{ff}$, $\sigma_{fs}^{}=
\sigma_{sf}$, and 
$\sigma_{ss}^{}$ denote the elements of the covariance matrix $\Sigma_\nu$. Appendix~\ref{sect:example:noidnet} shows that there exist two parameter vectors $\theta^I := 
\left( 
\alpha_f^{I},
\alpha_s^{I},
1,
\beta_{s}^{},
\sigma_{ff}^{I},
\sigma_{fs}^{I},
\sigma_{ss}^{I}
\right)^{\prime} \not=
\theta^{II} := 
\left( 
\alpha_f^{II},
\alpha_s^{II},
1,
\beta_{s}^{},
\sigma_{ff}^{II},
\sigma_{fs}^{II},
\sigma_{ss}^{II}
\right)^{\prime}
$ such that all observable second moments are the same; hence in this case the mapping from
the model parameters to
observable second moments cannot be injective and the model parameters are not identified from observed second moments. In this example $\alpha_f^{I} = \alpha_f^{II}=0$. This implies that the fast coordinate follows a random walk and does not provide any information on the 
parameter $\alpha_s$, that is on how $\beta^{\prime}y_t$ affects $\Delta y_{t}^{s}$, $t \in 2 \mathbb{Z}$.
However, in this paper we prove that identifiability holds for a so called generic subset of $\Theta$. Note that a set $\Theta_I \subset \Theta$ is called \textit{generic} in $\Theta$, if it contains a subset that is open and dense in $\Theta$.\\ 
Let $\Theta_{I} := \left(G \cap \Theta_1 \right) \times \Theta_2$, where $G\subset \mathbb R^{n^2 p}$ is defined in Assumption \ref{a: g-ident} below. In this paper we show firstly that $\Theta_I$ is generic in $\Theta$ (see Section \ref{sec: topological props}) and secondly that the set of high frequency systems corresponding to $\Theta_I$ is identifiable from the observable second moments (see Section \ref{sec: super idea}). Or formally, we show that 
\begin{align}
    \pi: \theta \mapsto \tilde \gamma \label{eq: def pi}
\end{align}
is injective on $\Theta_I \subset \Theta$. \\
Finally, in terms of identifiability, we may suppose without loss of generality that $\beta$ is known. For instance \citet{miller2016conditionally} or \citet{chambers2020frequency} propose estimators, accounting for stock and flow variables, respectively.
The estimators of $\beta$ scaled by $T$ weakly converge to 
a random variable bounded in probability. Hence, e.g. by \citet{white2001asymptotic}, the estimator is weakly consistent. 
By \citet{GABRIELSEN1978}  the matrix of cointegrating vectors $\beta \in \mathbb{R}^{n \times r}$ is identified from mixed frequency observations given the assumptions imposed in \citet{chambers2020frequency} or \citet{miller2016conditionally}.
These assumptions are only posed on the stochastic properties of the high frequency innovations $(\nu_t)_{t \in \mathbb Z}$ and therefore do not restrict our results on the genericity of the identifiability conditions from Section~\ref{sec: topological props}. If strong consistency could be established for some estimator of $\beta$, the results of 
\citet{DeistlerSeifert1978} apply and $\beta$ is identified. 
\subsection{Generic Identifiability and Topological Properties of the Parameterspace}\label{sec: topological props}
In this section we define the conditions that we need for identifiability and prove that these conditions result in a generic subset of the parameterspace. Define a set $G \subset \mathbb R^{n^2 p}$ by the following assumptions: 
\begin{assumption}[g-Identifiability Assumptions]\label{a: g-ident}\ \\[- 2.2em]  
\begin{itemize}
\item[\textbf{(I1)}] $\rank \mathcal{A}_p = n$. 
\item[\textbf{(I2)}] $\rank \Gamma(t) = np$ where $\Gamma(t) = \E( X_{t + 1} X_{t + 1}')$. 

\item[\textbf{(I3)}] The eigenvalues of $\mathcal{A}$ are of the form: $(1, ..., 1, \lambda_{n-r+1}, ..., \lambda_{np})$ where $|\lambda_{j}| < 1$ and $\lambda_i \neq \lambda_j$ for $i \neq j$ with $i,j = n-r+1, ..., np$.
\item[\textbf{(I4)}] For non-unit eigenvalues $\lambda_i \neq \lambda_j$ it follows that $\lambda_i^N \neq \lambda_j^N$. 
\item[\textbf{(I5)}] For all eigenvalues $\lambda$ of $\mathcal{A}$ smaller than one, it holds that $1 + \lambda + \cdots + \lambda^N \neq 0$ or $v_1$ consisting of the first $n$ elements of the eigenvector $v$ of $\mathcal{A}$ corresponding to $\lambda$, it holds that $\beta'v_1 \neq 0$.
\item[\textbf{(I6)}] The pair $(S_{n_f}^{(1)}, A)$ is observable, where $S_{n_f}^{(1)}$ is defined in equations (\ref{eq: hf ss obs_FD_2}), (\ref{eq: C_b structural matrix definition}) and $A$ is defined in equation (\ref{eq: hf ss trans_FD}). 
\end{itemize}	
\end{assumption}
%
%
%
%
%
Assumption 
(I2) already follows from $\Sigma>0$.
Recall that $\Theta_{I} = \left(G \cap \Theta_1 \right) \times \Theta_2$. 
%
%
%
%
%
These assumptions are similar to the stationary case considered in \citet{Felsenstein2014, AndersonEtAl2016a, AndersonEtAl2016b}. 
There, the stability condition defines an open set $\Theta' \subset \mathbb R^{n^2 p }$. 
We also have a corresponding set $G'$ defining the identifiability conditions for the stationary case, which is generic in $\mathbb R^{n^2 p }$.
%
Then, the intersection $\Theta' \cap G'$ is generic in $\Theta '$. 
However, in the integrated case, where unit roots occur, the situation is more intricate since neither $\Theta_1$ nor $G$ is open in $\mathbb R^{n^2 p}$. This follows from the fact that for a process with $n-r$ common trends, the $n-r$ eigenvalues of $\mathcal A$ in (\ref{eq: hf ss trans}) are equal to one [note that the eigenvalues of $\mathcal{A}$ are the reciprocals of the zeros of $a(z)$]. The following Theorem \ref{thm: topo thm} implies that the identifiability conditions are generically fulfilled in $\Theta$:
\begin{theorem}\label{thm: topo thm}
Let $\Theta_1$ be endowed with the Euclidean norm $d$. The set $\Theta_1 \cap G$ is open and dense in $\Theta_1$.
\end{theorem}
Since genericity is a topological property, it also holds for the homeomorphic parameterspace corresponding the vector error correction representation in (\ref{eq: VECM hf}) defined by Assumption \ref{ass: Theta}. 
%
%
%
%
%
%
%
%
%
%
%
%
%
%
%
%
%

\section{Generic Identifiability}
\label{sec: super idea}
%
%
In this section, we first define a canonical state-space representation for the blocked process running on $t \in N \mathbb Z$. We prove that this representation is minimal under our identifiability conditions. Then under an additional assumption on the lag order $p$, we show that the high frequency parameters are generically identifiabile. The proofs of minimality and identifiability make use of the canonical representation.\\
We follow \citet{HansenJohansen1999} and obtain from (\ref{eq: VECM hf}) the following state-space system for $\beta^{\prime} y_t$ and first differences of $y_t$, that is 
$\Delta y_t = y_t - y_{t-1}$. Then,
\begin{align}
\underbrace{\begin{pmatrix}
	\beta^{\prime} y_t \\
	\Delta y_t \\
	\vdots \\
	\Delta y_{t-p+2}
	\end{pmatrix}}_{\underline{x}_{t+1} \in \mathbb{R}^{r + n(p-1)}} &= 
\underbrace{\begin{pmatrix}
	\beta^{\prime} \alpha+ I_r &  \beta^{\prime} \Phi_1 & \cdots & \cdots & \beta^{\prime} \Phi_{p-1}   \\
	\alpha &   \Phi_1 & \cdots &\cdots &  \Phi_{p-1} \\
	0_{n \times r} & I_n &  & & 0_{n \times n} \\
	\vdots & & \ddots & &  \vdots \\
	& & & I_n & 0  
	\end{pmatrix} }_\text{$A \in \mathbb{R}^{r + n(p-1) \times r + n(p-1)} $}
\underbrace{\begin{pmatrix}
	\beta^{\prime} y_{t-1} \\
	\Delta y_{t-1} \\
	\vdots  \\
	\Delta y_{t-p+1}
	\end{pmatrix}}_{\underline{x}_t} + 
\underbrace{\begin{pmatrix}
	\beta^{\prime}  \\
	I_n \\
	0 \\
	\vdots \\
	0
	\end{pmatrix}}_\text{$B$}
\nu_t \label{eq: hf ss trans_FD} \\
%
%
%
\underbrace{\begin{pmatrix}
	\beta^{\prime} y_t \\
	\Delta y_{t} \\
	\end{pmatrix}}_{\in \mathbb{R}^{r+n}} 
&
= 
\underbrace{\begin{pmatrix}%
	\beta^{\prime} \alpha+ I_r &  \beta^{\prime} \Phi_1 & \cdots & \cdots & \beta^{\prime} \Phi_{p-1}   \\
	\alpha &   \Phi_1 & \cdots &\cdots &  \Phi_{p-1}   \\
	\end{pmatrix} }_{C \in \mathbb{R}^{r + n \times r + n(p-1)}}
\begin{pmatrix}
	\beta^{\prime} y_{t-1} \\
	\Delta y_{t-1} \\
	\vdots  \\
	\Delta y_{t-p+1}
	\end{pmatrix} + 
\underbrace{\begin{pmatrix}
	\beta^{\prime}  \\
	I_n \\
	\end{pmatrix}}_\text{$D \in \mathbb{R}^{r+n \times n}$} 
\nu_t \ . \label{eq: hf ss obs_FD}
\end{align}  
By $m := r + n(p-1)$, we denote the dimension of $\underline{x}_t$. As we will see later, given that our identifiability conditions hold $m$ is also the McMillan degree of $(\tilde y_t)_{t \in N\mathbb Z}$.
\\
According to the observational scheme, the slow variables $y_t^s$ are observed only every $N$-th period. We derive state-space representations for the processes (\ref{eq:tildeyt_FDstock}) and (\ref{eq:tildeyt_FD_flow}) running on $t \in N \mathbb Z$:\\
\textbf{1. Case: Stock Variables:} We define a new state vector $x_{t+1}$ in the following way, with the condition that $p \geq N + 2$: 
%
%
{\footnotesize
	\begin{align}
	\underbrace{\begin{pmatrix}
		u_t^\mathcal{S} \\
		\Delta_N y_t \\
		\vdots \\
		\Delta_N y_{t - N + 1} \\
		\Delta y_{t - N} \\
		\vdots \\
		\Delta y_{t-p+2}
		\end{pmatrix}}_{x_{t+1} \in \mathbb{R}^{r + n(p-1)}  } &= 
	\underbrace{\begin{pmatrix}
		I_r & 0 & 0 & \dots & &  \\
		0   & I_n &I_n & \dots &I_n & 0 & \dots & \\ 
		0   & 0 &I_n & \dots & \dots & I_n & 0 &\dots \\
		\vdots & &  & \ddots & &  & \ddots & &  \\
		\vdots   &  & &  & I_n & \dots & \dots & I_n & 0 \dots \\	
		\vdots   &  & &  \dots & 0 & I_{n} & 0 & \dots \\ 
		\vdots   &  & &   &  & & \ddots &  &  \\
		\end{pmatrix} }_\text{$c \in \mathbb{R}^{r + n(p-1) \times r + n(p-1)} $}
\begin{pmatrix}
		u_t^\mathcal{S} \\
		\Delta y_t \\
		\vdots  \\
		\Delta y_{t-p+2}
		\end{pmatrix}  \ . \label{eq:cmatrix_FD}
	\end{align}  }
By iterating the system (\ref{eq: hf ss trans_FD}), (\ref{eq: hf ss obs_FD}), we get the non-miniphase system (in the sense that the transfer-function is not causally invertible as the input dimension exceeds ($Nn$) the output dimension ($\tilde n$), noting that $\Sigma_\nu > 0$):
{\footnotesize 
\begin{align}
%
x_{t+1} & = 
\underbrace{c\underbrace{ A^N }_{:= A_b} c^{-1}}_{A_{b, c}} x_{t-N+1} + 
\underbrace{c B_{b}}_{B_{b, c}}  \nu_t^b  \label{eq: hf ss trans_FD_2}\\
%
%
%
%
%
\tilde y_{t} & = 
\underbrace{\underbrace{
S_{\zeta}
	A^N}_{:= C_b} c^{-1}}_{:= C_{b, c}}
x_{t-N+1} + D_b
\nu_t^b \ , \label{eq: hf ss obs_FD_2} \\
& \mbox{where} \nonumber \\
& S_{\zeta} := 
\begin{matrix}
 \quad  \ \ \ (r \times m) \{ \\
S_{n_f}^{(1)}   \ (n_f \times m) \{ \\
S_{n_s}^{(1)}   \ (n_s \times m) \{ \\
S_{n_f}^{(2)}  \  (n_f \times m) \{ \\
   \vdots \\
 S_{n_f}^{(N)}   \ (n_f \times m) \{
\end{matrix}
\begin{pmatrix}%
	I_r & 0 & \cdots &   &  &  & 0  \ \\
	0 & \left( I_{n_f} , 0 \right) & \cdots & \left( I_{n_f} , 0 \right)  & 0  & \cdots  & 0\\
	0 & \left(  0 , I_{n_s} \right) & \cdots & \left( 0 , I_{n_s} \right)  & 0  & \cdots  & 0\\
	0 & 0 & \left( I_{n_f} , 0  \right) & \cdots  &  \left( I_{n_f} , 0  \right) & & \\
	& & & \ddots & & \ddots & \vdots \\
	&  &  & \left( I_{n_f} , 0  \right) & \cdots &  \left( I_{n_f} , 0  \right)  & 0 \\
	\end{pmatrix}  \ , \nonumber \\[1.5em]
& C_b = \begin{pmatrix}
\begin{pmatrix}I_r & 0 & \cdots &  0 \end{pmatrix} A^N\\
  S_{n_f}^{(1)}      A^N\\
 S_{n_s}^{(1)}       A^N\\
 S_{n_f}^{(2)}  A^N\\
   S_{n_f}^{(3)} A^N\\
\vdots  \\
S_{n_f}^{(N)}  A^N \\
\end{pmatrix} 
= 
\begin{pmatrix}
\begin{pmatrix}I_r & 0 & \cdots & 0 \end{pmatrix} A^N\\
  S_{n_f}^{(1)}      A^N\\
 S_{n_s}^{(1)}       A^N\\
 S_{n_f}^{(1)}  A^{N - 1}\\
   S_{n_f}^{(1)} A^{N - 2}\\
\vdots  \\
S_{n_f}^{(1)}  A  \\
\end{pmatrix} \ , \ \
\nu_t^b :=
\begin{pmatrix}
\nu_{t} \\
\vdots \\ 
\vdots \\
\nu_{t-N+1} 
\end{pmatrix}
\in \mathbb{R}^{Nn} .
\label{eq: C_b structural matrix definition}
\end{align}
%
}
The matrices $B_{b, c} \in \mathbb{R}^{r+n(p-1) \times Nn} $ and
$D_b \in \mathbb{R}^{r+n \times Nn}$ are obtained from $B$ and $A$.\\
\textbf{2. Case: Flow Variables:}
Next, we obtain the state vector $x_{t+1}$ for the flow case. Note that $y_{t-j} = y_t - \sum_{\ell=1}^{j} \Delta y_{t-\ell}$, such that
$\sum_{j=0}^{N-1} y_{t-j} = 
\sum_{j=0}^{N-1} \left( y_t - \sum_{\ell=1}^{j} \Delta y_{t-\ell} \right)
= N y_t - (N-1) \Delta y_{t-1} - \cdots - \Delta y_{t-N+1}
$. Analogously to equation (\ref{eq:cmatrix_FD}), this yields for $p \geq 2N+1$ that
%
%
%
%
{\scriptsize
	\begin{align}
	\underbrace{\begin{pmatrix}
		\beta^{\prime} \sum_{j=0}^{N-1} y_{t-j} \\
		\Delta_N^\Sigma y_t \\
		\Delta_N y_{t - 1} \\
		\vdots \\
		\Delta_N y_{t - N + 1} \\
		\Delta y_{t - N} \\
		\vdots \\
		\Delta y_{t-p+2}
		\end{pmatrix}}_{x_{t+1} \in \mathbb{R}^{r + n(p-1)} } &=
	\underbrace{\begin{pmatrix}
		N I_r & -(N-1) \beta^{\prime}  & -(N-2) \beta^{\prime}  & \cdots & 
		-\beta^{\prime} & 0 & \cdots &  \\
		0      & I_n & \cdots  & I_n & - I_n & \cdots & -I_n & 0  \\ 
		0      & 0   & I_n     & \cdots & I_n & 0 & \cdots  & \\
		\vdots &     &         & \ddots & &  \ddots & &   \\
		\vdots &     & &       & I_n & \cdots &  I_n & 0  \\	
		\vdots &     & &\cdots & 0 & I_{n} & 0 & \cdots \\ 
		\vdots &     & &  & &    & \ddots &   \\
		\end{pmatrix} }_\text{$c \in \mathbb{R}^{r + n(p-1) \times r + n(p-1)} $}
       \begin{pmatrix}
		u_{t}^\mathcal{S} \\
		\Delta y_{t} \\
		\vdots  \\
		\Delta y_{t-p+2}
		\end{pmatrix} \ . \label{eq:cmatrix_FD_flow}
	\end{align}  }
%
%
%
%
%
%
%
\noindent We use the same notation for $\tilde y_t$, $x_t$, $c$ for both cases. With this notation, we obtain the following state-space representation for blocked process in the flow case:
{\footnotesize
\begin{align}
x_{t+1}
& = \underbrace{c A_b c^{-1} }_{A_{b,c}}  x_{t-N+1}
+ \underbrace{c B_b}_{B_{b,c}}  
\nu^b_t \label{eq: mixedf ss trans_FD_6} \\
%
%
%
\tilde{y}_t  
& = 
\underbrace{\underbrace{
S_{\zeta}	A^N
}_{C_{b} \in \mathbb{R}^{\tilde{n} \times m } } c^{-1} }_{C_{b, c}
\in \mathbb{R}^{\tilde{n} \times m }
}
x_{t-N+1}  +  D_{b,c}  
\nu_t^b \ , \label{eq: mixedf ss obs_FD7} \\
& \text{ where } \nonumber \\
& S_{\zeta} =
\begin{pmatrix}
		N I_r & -(N-1) \beta^{\prime}  & -(N-2) \beta^{\prime}  & \cdots & 
		-\beta^{\prime} & 0 & \cdots & & \\
		0   & I_n & \cdots & I_n & - I_n & \cdots & -I_n & 0 &  \cdots  \\ 
		0   & 0 & (I_{n_f}, 0) & \cdots & (I_{n_f}, 0) & 0 & \cdots  &  \\ 
		&&&& \ddots  &&&& \\
		0 &   & 0 & (I_{n_f}, 0) & \cdots & (I_{n_f}, 0) & 0 & \cdots  
		\end{pmatrix} . \nonumber 
\end{align} 
}
The matrix $D_{b,c} \in \mathbb{R}^{\tilde{n} \times Nn} $ follows from $D_b$, the matrix $c$ and the selection of the corresponding rows resulting in $\tilde{y}_t$. \\
%
%
%
%
\textbf{3. Case: Mixed Case}: Consider the case where we have slow stock as well as slow flow variables:
For example, if $\left( y_t \right)$ is a three-dimensional process, where 
$n_f=1$, $n_s=2$, $N=2$, $c_1=I_2$, and $c_2 = \begin{pmatrix}
      0 & 0 \\
      0 & 1
\end{pmatrix}$ in equation~(\ref{eq: def w_t}).
Then $\beta^{\prime} 
\left(y_t^{f \prime},w_t^{\prime} \right)^{\prime}
$ is (in general) not stationary. 
%
However, in special cases, such as separate cointegrating relationships among the slow flow variables only, or among the slow stock and fast variables only, etc. we can proceed similarly to the flow case. In the following we only consider the stock or the flow case. \\
%
%
%
%
%
%
%
%
%
The problem with the systems considered above is that the inputs $\nu^b_t$ are not the innovations of $\tilde{y}_t$. However, from the stable miniphase spectral factorisation, we only obtain transfer functions corresponding to systems in innovation form \citep[see, e.g.,][Chapter~7]{ScherrerDeistler2018book}. The following Theorem \ref{thm: th1aggregated} is the first step for obtaining a canonical state-space representation for the blocked process. A minimal state-space representation is called ``canonical'' if its parameters are uniquely determined from the transfer function. We introduce the following notation for specific subspaces of $L^2(\Omega, \mathcal A, P)$, the space of square integrable random variables on the underlying probability space $(\Omega, \mathcal A, P)$:
\begin{align*}
\mathbb{H}(y) &:= \overline{\spargel}(y_{it} \mid t \in \mathbb{Z}, i = 1,...,n) \\
\mathbb{H}_t(y) &:= \overline{\spargel}(y_{is} \mid s \leq t, i = 1, ..., n) \\
{}_{N}\mathbb{H}(y) &:= \overline{\spargel}(y_{it} \mid t \in N\mathbb{Z}, i = 1, ..., n) \\
{}_{N}\mathbb{H}_t(y) &:= \overline{\spargel}(y_{is} \mid s \leq t \ \mbox{and} \ s \in N\mathbb{Z}, i = 1, ..., n) \ , 
\end{align*}
where $\overline{\spargel}(\cdot)$ denotes the closed span 
and $\proj(v \mid U)$ the projection of $v$ on a closed subspace $U$ of $L^2$.
%
%
%
%
%
%
%
%
%
%
%
\begin{theorem}\label{thm: th1aggregated}
Suppose that \textcolor{red}{Assumption \ref{ass: Theta}} holds. Consider the blocked process $(\tilde y_t)_{t \in N \mathbb Z }$ 
and set
\begin{align*}
	  s_{t-N+1} &:= \proj(x_{t-N+1} \mid {}_N \mathbb{H}_{t-N}(\tilde y)) \\
	  \tilde \nu_t  &:= \tilde y_t - \proj(\tilde y_t \mid {}_N \mathbb{H}_{t-N}(\tilde y)) \ .
\end{align*}
	Then there exists $\tilde B_c \in \mathbb R^{np \times \tilde n}$ such that 
	\begin{align}
	    s_{t+1} &= A_{b,c}s_{t-N+1} + \tilde B_c \tilde \nu_t  \label{eq: beta stock inno trans}\\ 
    	\tilde y_t &=  C_{b,c} s_{t-N+1} + \tilde \nu_t \label{eq: beta stock inno obs}
	\end{align}
	is a miniphase and stable state-space representation of $( \tilde{y}_t )_{t \in N\mathbb Z}$, i.e. it is in innovation form.
\end{theorem}
We call the representation in (\ref{eq: beta stock inno trans}), (\ref{eq: beta stock inno obs}) \textit{canonical projection form} (CPF) of $\tilde{y}_t$. Note that the CPF provides an algorithm for computing the transfer function $\tilde k(\tilde z)$ of $(\tilde y_t)_{t \in N \mathbb Z}$ which corresponds to the Wold representation, where $\tilde z := z^N$.\\ 
Next we show that the system (\ref{eq: beta stock inno trans}) and (\ref{eq: beta stock inno obs}) is observable and controllable and therefore minimal \citep[see, e.g.,][Theorem~2.3.3]{HannanAndDeistler2012} for all $\theta \in \Theta_I$. 
\begin{theorem}\label{thm: Minimal}
For $\theta \in \Theta_I$, the system (\ref{eq: beta stock inno trans}) and (\ref{eq: beta stock inno obs}) is minimal. 
\end{theorem}
By Theorem \ref{thm: Minimal}, we know that the McMillan degree of the transfer function of the blocked process $(\tilde y_t)_{t \in N \mathbb Z}$ corresponding to an underyling high-frequency VECM is $m = r + n (p - 1)$. This will be used in the proof of the subsequent Theorem \ref{thm: identifiability in steps}, where we can relate an arbitrary minimal realisation $(\bar A_{b,c}, \bar B_{b, c}, \bar C_{b, c})$ of the transfer function $\tilde k(\tilde z)
= \left(\bar{C}_{b,c} \left( I_{m} \tilde{z}^{-1}
- \bar{A}_{b,c}
\right) 
\bar{B}_{b,c} + I_{\tilde{n}}
\right)
$ (where $\tilde z := z^N$) to the CPF $(A_{b, c}, \tilde B_{c}, C_{b, c})$. The minimal realisation $(\bar A_{b,c}, \bar B_{b, c}, \bar C_{b, c})$ can be either obtained by the spectral factorisation and e.g. the echelon realisation from the Hankel matrix of the transfer function \citep[see e.g.][Theorem 2.6.2]{HannanAndDeistler2012} or directly from the Hankel matrix of the observed second moments \citep[see, e.g.][Proof of Theorem 8]{AndersonEtAl2016a}. In the next step we relate the CPF to the underlying VECM/VAR -- exploiting the fact that the parameters $\theta$ of the underyling VECM reappear in the CPF.\\
Finally, we show that the parameters of the high frequency system are generically identifiable from the observed second moments, i.e. from $\tilde \gamma$.  
\begin{theorem}[Generic-Identifiability: Flow or Stock Case]\label{thm: identifiability in steps}
Let
$p \geq N+2$ for stock case or $p \geq 2N+1$ for the flow case. Then, 
\begin{enumerate}
    \item The mapping, $\pi$ in equation (\ref{eq: def pi}) which attaches the second moments of $(\tilde y_t)_{t \in N \mathbb Z}$ to the high frequency parameters $\theta$ is injective on $\Theta_I$.
    \item Its inverse, $\pi^{-1}$, is continuous on $\pi(\Theta_I)$.
\end{enumerate}
\end{theorem}
%
%
%
%
%
Since by Theorem \ref{thm: topo thm}, $\Theta_I$ is a generic subset of $\Theta$, we say that $\theta$ is generically identifiable from the observed autocovariance function $\tilde \gamma$. Theorems \ref{thm: Minimal} and \ref{thm: identifiability in steps} imply that the representation (\ref{eq: beta stock inno trans}), (\ref{eq: beta stock inno obs}) is indeed canonical on $\pi(\Theta_I)$. Since the second moments of $(\tilde y_t)$ can be consistently estimated from the data under mild conditions, by the continuity of $\pi^{-1}$ it follows that we have a consistent estimator for $\theta$. The mapping $\pi^{-1}$ is also called \textit{realisation procedure}, since we realise the system parameters from the external characteristics of the data, i.e. the second moments, the spectrum or the transfer function respectively.\\ 
%
%
%
Finally, we consider the question whether  $\pi^{-1}(\pi (\Theta_I)) = \Theta_I$. This is important to ensure that 
outside that the identified parameter set 
$\Theta_I$ there are no elements, say
$\theta_{\neg I}$, which result in the same observable second moments as some 
$\theta_{} \in \Theta_I$:
\begin{theorem}\label{thm: Manfreds Question}
For all $\theta_{\neg I} \in \Theta \setminus \Theta_I$ there exists no $\theta \in \Theta_I$ such that $\pi(\theta_{\neg I}) = \pi(\theta) = \tilde \gamma$.
\end{theorem}

\section{Deterministic Terms}
\label{sect:deter}
This section investigates the VECM  
\begin{align}
\label{eq:vecm:det1}
\Delta y_t  = \mu_0 + \mu_1 t + \Pi y_{t-1} + \sum_{j = 1}^{p-1} \Phi_j \Delta y_{t-j} + \nu_t \ ,   \qquad \nu_t \sim WN(\Sigma_\nu)  \  ,
\end{align}
containing the deterministic terms $\mu_0$ and $\mu_1 t$. The five cases following from (\ref{eq:vecm:det1}), namely  
``$H_2(r)$, $H_1(r)$, $H_1^*(r)$,
$H(r)$, and $H^*(r)$'', are
obtained and defined in \citet{Johansen1995}[page 81 and 
our Appendix~\ref{appa:deterministic}]. Recall that the cointegrating vectors $\beta$ can be identified from mixed frequency data \citep[see][]{chambers2020frequency}. For high frequency data $(\Delta y_t)_{t \in \mathbb Z}$ and $(\beta'y_t)_{t \in \mathbb Z}$ we can compute the expectations
$\E \Delta y_t$ and $\E \beta^{\prime} y_t$, while for mixed frequency case  we get $\E \beta^{\prime} y_t$ and $\E \Delta_N y_t = \mathbb{E} \left(y_t - y_{t-N} \right)$
for the stock case and 
$\E \beta^{\prime} w_t = 
\E \sum_{j=0}^{N-1} \beta^{\prime} y_{t-j}$
and $\E \Delta_N^{\Sigma} y_t = 
\E \sum_{j=0}^{N-1} \Delta_N y_{t-j}$ for the flow case, respectively.
To identify the deterministic terms in 
(\ref{eq:vecm:det1}) we can proceed as follows:
\begin{enumerate}
\item Remove deterministic trends from 
$\left( \beta^{\prime} y_t \right)_{ t \in N \mathbb{Z} } $ and $\left( \Delta_N y_t  \right)_{ t \in N \mathbb{Z} } $
in the stock case or 
$\left( \beta^{\prime} w_t \right)_{ t \in N \mathbb{Z} }$
and $\left( \Delta_N^{\Sigma} y_t \right)_{ t \in N \mathbb{Z} } $ for the flow case
[given $\beta$], such that
$\E \beta^{\prime} y_t =0$ and $\E \Delta_N y_t = 0$, 
$\E \beta^{\prime} w_t = 0 $, and $\E \Delta_N^{\Sigma} y_t = 0$.

\item
Apply Theorem~\ref{thm: identifiability in steps} to obtain the parameters $\theta$.
\item Given the parameters $\theta$, 
Appendix~\ref{appa:deterministic} shows that the parameters
$\mu_0$ and $\mu_1$ can be identified from moments following from data observed.
\end{enumerate}
This results in:
\begin{theorem}
\label{thm:identify:mu0mu1}
Under the assumptions of Theorem~\ref{thm: identifiability in steps} 
the parameters $\mu_0$ and $\mu_1$ can be identified generically from observable first and second moments.
\end{theorem}
\begin{proof}
See Appendix~\ref{appa:deterministic}.
\end{proof}
\section{Conclusion}
\label{sect:conc}
%
%
%
%
In this paper, we generalise the results on identifiability from mixed frequency data in \cite{AndersonEtAl2016a, AndersonEtAl2016b} obtained for stationary VAR-systems to the case of unit-roots and cointegrating relationships. As is well known these systems have also a {\em vector error correction representation}. The corresponding parameterspaces are homeomorphic.\\
We commence from a solution of the (unstable) VAR system on the integers $\mathbb Z$ \citep[see][for the existence of such a solution]{bauerwagner2012}. Then we take differences at lag $N$ (which is the sampling rate of the slow/aggregated process) and stack these to what we call the ``blocked process''. In addition, the blocked process also contains the stationary process $\beta'y_t$, where $\beta$ is the matrix of cointegrating relationships. This matrix is identified from mixed frequency data as already shown in \cite{chambers2020frequency}. This blocked process is stationary and contains all relevant differences of the observations. \\
The contribution of this paper can be seen as an extension of the results in \cite{chambers2020frequency}, by proving that also the remaining parameters of the vector error correction model (i.e. besides $\beta$) are (generically) identified from mixed frequency observations.\\ 
%
The identifiability proof consists of two steps: In the first step, we derive a state-space representation of the blocked process (``the canonical projection form'') which is minimal, in innovation form (both, for the stock and the flow case) and unique. In the second step, we derive an algorithm, that retrieves the parameters of the underlying high frequency system from the parameters of the canonical projection form.\\ %
We show that the conditions (Assumption \ref{a: g-ident}) which are sufficient for identifiability are generic in the parameterspace. This is more intricate than in the stationary case, since the parameterspace is not an open subspace of the Euclidean space, due to the fact that we allow for unit roots. Since the VECM and the VAR parameterspaces are homeomorphic, the genericity result holds for both.\\
Finally, we show that all common cases of deterministic terms in the VECM can be reduced to the case of non-deterministic terms.
%
%
%
%
%
%
%
%
%
%
%
%
%
\bibliographystyle{apalike} 
%
\bibliography{references.bib}
\appendix
\renewcommand{\theequation}{A.\arabic{equation}}

\setcounter{equation}{0}


\section*{Conflict of Interest Statement }
On behalf of all authors, the corresponding author states that there is no conflict of interest. 

\section*{Acknowledgements}

{{The authors would like to thank 
the participants of the Economics Research Seminar of the University of Regensburg, the CFE 2020 conference, the
3rd Italian Workshop of Econometrics and Empirical Economics (IWEEE 2022), and the 10th Italian Congress of Econometrics and Statistics (ICEEE 2023), as well as Christoph Rust for helpful comments that lead to improvement of the paper. 
The authors gratefully acknowledge financial support from the Austrian Central Bank under Anniversary Grant No. 18287 and the DOC-Fellowship of the Austrian Academy of Sciences (ÖAW).
Manfred Deistler and Leopold Sögner acknowledge support by the Cost Action HiTEc - CA21163.
} }

\clearpage
\section{Moments Observed}
\label{sect:momentsobserved}
In {Section}~\ref{subsect: mixed frequency data and generic ident}, we call the moments $\E \tilde y_{t+h} \tilde y_t'$ in equation (\ref{eq: def tilde gamma}) also {\em observed second moments} since they can be consistently estimated from the data as we will argue in this section. Those moments in $\E \tilde y_{t+h} \tilde y_t'$ which involve only $N$-differences from the data do not require any further discussion. We are left with considering moments that involve $\beta$.\\ 
%
%
In the following, if $N$ is the sampling rate of the slow variables and $T$ is the number of time observations of high-frequency data, we denote by $T_s$ the number of time observations available for the slow/ aggregated coordinates which is approximately $T/N$ and by $\mathbb I_s$ the index set of slow/ aggregated time observations, e.g. for $N=3$ we could have $\mathbb I_s = \{3, 6, 9, ...,T\}$. 

In \citet{miller2016conditionally, chambers2020frequency} consistent estimators $\hat \beta_{T_s}$ that converge in probability to $\beta$ are provided for the mixed frequency case. In both articles, the corresponding estimator $\hat \beta_{T_s}$ converges weakly to a product of the inverse of a functional of Brownian motions and a stochastic integral. 
By applying \citep[][{Lemmata 4.5 and 4.6}]{white2001asymptotic}, we observe that 
\begin{align*}
(\widehat{\beta}_{T} - \beta) = O_p(T_s^{-1}) \ .  
\end{align*}
Next we consider the high-frequency moments $\E \beta^{\prime} y_t$, $\E  \beta^{\prime} y_t \Delta y_s^{\prime}$, and $\E  \beta^{\prime} y_t y_s^{\prime} \beta^{}$, $s,t \in \mathbb{Z}$. The following calculations work in the same way for the mixed frequency counterparts. We can estimate $\E \beta^{\prime} y_t$ by $\frac{1}{T_s}\sum_{t \in \mathbb I_s} \widehat{\beta}_{T_s}^{\prime} y_t$. We impose the standard assumptions that a weak law of large number and a functional central limit theorem can be applied to properly scaled terms. 
For a ``Low-Frequency Invariance Principle'' see, e.g., 
\citet{miller2016conditionally}[Section~2.3].
Let $B(\tau)$ denote some $n$-dimensional Brownian motion. 
Define $\widehat u_t^{\mathcal S} := \widehat \beta_{T_s}' y_t$ and $\widehat u_t^{\mathcal F} := \widehat \beta_{T_s} \sum_{j = 0}^{N-1}y_{t-j}$. 
This results in 
\begin{align}
\label{eq:a1}
T_s^{-1} \sum_{t \in \mathbb I_s} \widehat u_t^{\mathcal S} &= T_s^{-1} \sum_{t \in \mathbb I_s} \widehat{\beta}_{T_s}^{\prime} y_t 
= T_s^{-1} \sum_{t \in \mathbb I_s} {\beta}^{\prime} y_t + T_s^{-1} \sum_{t \in \mathbb I_s} \left( \widehat{\beta}_{T_s} - \beta \right)^{\prime} y_t \nonumber \\
& = \E {\beta}^{\prime} y_t + o_p(1) +  T_s^{1/2}(\widehat{\beta}_{T_s} - \beta)'\underbrace{T_s^{-3/2}\sum_{t \in \mathbb I_s} y_t}_{\Rightarrow \int_0^1 B(\tau) d \tau}\nonumber \\
&= \E \beta'y_t + o_p(1) + O_p(T^{-1/2}) O_p(1) \ . 
\end{align}
Next, we look at the second moments for the stock case appearing in $\tilde \gamma$, i.e. second moments of the form $\E u_t^{\mathcal S} \Delta_N y_s$. Observe that we have
\begin{align}
\label{eq:a2}
T_s^{-1} \sum_{t \in \mathbb I_s} \widehat u_t^{\mathcal S}  \Delta y_s^{\prime}  & = T_s^{-1} \sum_{t \in \mathbb I_s} \widehat{\beta}_{T_s}^{\prime} y_t \Delta y_s^{\prime}  \nonumber \\
& = T_s^{-1} \sum_{t \in \mathbb I_s} \beta' y_t \Delta y_s' + \left( \widehat{\beta}_{T_s} - \beta \right)' \underbrace{T_s^{-1} \sum_{t \in \mathbb I_s}  y_t \Delta y_s^{\prime}}_{\Rightarrow \int_0^1 B(\tau) dB(\tau)'}  \nonumber \\
& = \E u_t^{\mathcal S} \Delta y_s' + o_p(1) +  O_p(T_s^{-1}) O_p(1) \ . 
%
%
%
\end{align}
It follows immediately that also 
\begin{align*}
    T_s^{-1} \sum_{t\in \mathbb I_s} \hat u_t^{\mathcal S} \Delta_N y_s \xrightarrow{P} \E u_t^{\mathcal S} \Delta_N y_s^{\cdot} \ , 
\end{align*}
where ``$\xrightarrow{P}$'' denotes convergence in probability and ``$\cdot$'' is shorthand for $s, f$ (``fast'' and ``slow'' coordinates). Finally we look the estimation of $\E u_t^{\mathcal S} {u_s^{\mathcal S}}'$ for $t - s \in N \mathbb Z$: 
\begin{align}
\label{eq:a3}
T_s^{-1} \sum_{t \in \mathbb I_s} \widehat u_t^{\mathcal S} (\widehat u_s^{\mathcal S})' & =  T_s^{-1} \sum_{t \in \mathbb I_s} \beta' y_t y_s' \beta  + T_s^{1/2}(\widehat \beta_{T_s} - \beta)' \underbrace{ T_s^{-2} \left(\sum_{t \in \mathbb I_s} y_t y_s'\right)}_{\Rightarrow \int_0^1 B(\tau) B(\tau)' d \tau} T_s^{1/2}(\widehat \beta_{T_s} - \beta) \nonumber \\
&= \E u_t^{\mathcal S} (u_s^{\mathcal S})' + o_{p}(1) +  O_p(T_s^{-1/2}) O_p(1) O_p(T_s^{-1/2}) \ .
%
%
%
%
%
%
\end{align}
We obtain analogous results for moments involving $u_t^{\mathcal F}$, by noting that $u_t^{\mathcal F} = \sum_{j = 0}^{N-1} \beta' y_{t-j}$. Hence, equations (\ref{eq:a1}) to (\ref{eq:a3}) demonstrate why the above moments can be considered as moments observed although $\beta$ has to be estimated.  
\clearpage
\section{Example for Non-Identifiability}
\label{sect:example:noidnet}

We consider a VECM where $n=2$, $r=1$, and $p=1$. The cointegrating vector is
$\beta = (1,\beta_{s})^{\prime}$, $\alpha = (\alpha_f, \alpha_s)^{\prime}$, 
and $\Pi = \alpha \beta^{\prime}$.
That is 
\begin{align}
	\Delta {y}_{t}
	&=  \Pi {y}_{t-1}   
	+ { \nu}_{t},  \quad  \nu_t \sim WN \left( \Sigma_{\nu} \right) \ ,  \  \text{such that} \ 
 \Sigma_{\nu} = \left( 
 \begin{array}{cc}
     \sigma_{ff} &  \sigma_{fs} \\
     \sigma_{sf} &  \sigma_{ss} 
 \end{array}
 \right) \ , \  \text{ and }
 \nonumber \\ 	
 y_t & = \underbrace{ 
 \left( I_n + \Pi \right)}_{
\mathcal{A}_1}
 y_{t-1} + \nu_t = 
 \left( 
\begin{array}{cc}
 \underbrace{ 
 1+ \alpha_f
 }_{a_{ff}} &
 \underbrace{ 
 \alpha_f \beta_{s}
 }_{a_{fs}} \\
 \underbrace{ 
 \alpha_s
 }_{a_{sf}} &
\underbrace{ 
 1+ \alpha_s \beta_{s}
 }_{a_{ff}}
 \end{array}
 \right) y_{t-1}
+ \nu_{t} \label{eq: VECM hf:p1} \ . 
%
%
\end{align}  
By Theorem~1 of \cite{AndersonEtAl2016a} the model parameters are not identified if and only if 
(i) $a_{fs}=0$, (ii) $a_{sf} 
+ \frac{\sigma_{sf}}{\sigma_{ff}} \left( a_{ss} - a_{ff} \right)=0$, 
and (iii) $a_{ss} \not=0$. 
$a_{fs}=0$ implies $\alpha_f=0$ [with $\beta_{s}=0$ we do not have a cointegrating relationship]. From $a_{sf} 
+ \frac{\sigma_{sf}}{\sigma_{ff}} \left( a_{ss} - a_{ff} \right)
= \alpha_s \left( 1
+ \frac{\sigma_{sf}}{\sigma_{ff}} \beta_{s} \right)
=0$, we conclude that $\alpha_s$ must be non-zero, otherwise $\left( y_t \right)$ is a white noise process. Hence, we get
$0= \left( 1
+ \frac{\sigma_{fs}}{\sigma_{ff}} \beta_{s} \right)$. The third constraint results in $0=1+ \alpha_s \beta_{s}$. Hence, by considering the second moments in levels and the arguments provided in the proof Theorem~1 of \cite{AndersonEtAl2016a}, we get identifyability on a generic set.
Unfortunately these moments are not ``observed moments'' since they cannot be estimated from data observed as in the stationary case.\\
By contrast $\mathbb{E} \Delta_N \tilde{y}_t \Delta_N \tilde{y}_{t-N  \ell}^{\prime}$, $\ell \in \mathbb{Z}$, are observed moments, where we consider stock variables only and let $\Delta_N \tilde{y}_t = y_t - y_{t-N}$. In addition, when $\beta$ is known [or can be consistently estimated as shown in \citet{miller2016conditionally,chambers2020frequency}] also the moments $\mathbb{E} \beta \tilde{y}_t \Delta_N \tilde{y}_{t- \ell}^{\prime}$, $t,t-\ell \in 2 \mathbb{Z}$, can be observed. Applying (\ref{eq: hf ss trans_FD}) and (\ref{eq: hf ss obs_FD}) to the current example yields
\begin{align}
\begin{pmatrix}
	\beta^{\prime} y_t \\
	\Delta y_t \\
	\Delta y_{t-1}
	\end{pmatrix}
&= 
\underbrace{\begin{pmatrix}
	\beta^{\prime} \alpha+ 1 &  0_{ 1 \times 2} & 0_{ 1 \times 2} \\
	\alpha &   0_{ 2 \times 2}  &  0_{ 2 \times 2} \\
    0_{ 2 \times 1} &  I_{ 2} & 0_{ 2 \times 2} 
 	\end{pmatrix} }_\text{$A \in \mathbb{R}^{r + 2n \times r + 2n} $}
\begin{pmatrix}
	\beta^{\prime} y_{t-1} \\
	\Delta y_{t-1} \\
    \Delta y_{t-2}
	\end{pmatrix} + 
\begin{pmatrix}
	\beta^{\prime}  \\
	I_n \\
	0
 \end{pmatrix}
\nu_t \ , \ \  t \in \mathbb{Z} \ . \label{eq: hf ss trans_FD:p1} 
\end{align}  
%
%
%
Next, by direct calculations we get
\begin{align}
	\Delta_2 {y}_{t}
	&=  \Delta y_t + \Delta y_{t-1} = \Pi {y}_{t-1}   
	+ { \nu}_{t} +  \Pi {y}_{t-2}   
	+ { \nu}_{t-1} \nonumber \\
  & = 
 \Pi \left( {y}_{t-2} + \Delta y_{t-1} \right)  
	+ { \nu}_{t}
+  \Pi {y}_{t-2}   
	+ { \nu}_{t-1} \nonumber \\
 & = 
2 \alpha \beta^{\prime} {y}_{t-2} + 
\alpha \beta^{\prime} 
\alpha \beta^{\prime} {y}_{t-2}   
	+ \Pi { \nu}_{t-1}
+ \nu_{t-1} + { \nu}_{t}
 \nonumber \\
 & =
 \underbrace{
\left( 
\begin{array}{c}
2 \alpha_f + \alpha_f^2 + \alpha_f \beta_{s} \alpha_s \\
2 \alpha_s + \alpha_f \alpha_s + \alpha^2_2 \beta_{s}
\end{array}
 \right) }_{ =: M_{\alpha,\beta}}
 \beta {y}_{t-2}   
	+ \Pi { \nu}_{t-1} + \nu_{t-1}
 + { \nu}_{t}  \ .
 \label{eq: VECM hf:p5}
\end{align}  
In addition, 
$\beta^{\prime} y_t 
=
\beta^{\prime} \left( y_{t-2} + \Delta_2 y_t \right)$ and 
$ \Delta y_{t-1}^f := y_{1t-1}  
= \alpha_f \beta y_{t-2} + (1,0)^{\prime} \nu_{t-1} $. 
This results in the state-space system
{\footnotesize
\begin{align}
\underbrace{
\begin{pmatrix}
	\beta^{\prime} y_t \\
	\Delta_2 y_{t} \\
	\Delta y_{t-1}^f
	\end{pmatrix} 
 }_{x_t }
&= 
\underbrace{\begin{pmatrix}
	1 + \beta^{\prime} M_{\alpha,\beta} & 0_{ 1 \times 2}
 & 0_{ 1 \times 2}
 \\
	M_{\alpha,\beta} &   0_{ 2 \times 2}  &  0_{ 2 \times 2} \\
 \alpha_f &    0_{ 1 \times 2} &   0_{ 1 \times 2} 
 	\end{pmatrix} }_\text{$\tilde{A} \in \mathbb{R}^{r + n +1 \times r + n +1} $}
\begin{pmatrix}
	\beta^{\prime} y_{t-2} \\
	\Delta_2 y_{t-2} \\
    \Delta y_{t-3}
	\end{pmatrix} + 
\underbrace{
\begin{pmatrix}
	\beta^{\prime} I_2 & \beta^{\prime} (I_2 + \Pi) \\
	 I_2 & I_2 + \Pi \\
	0 & (1,0)^{\prime}
 \end{pmatrix} }_{ {\tilde{B}} }
\begin{pmatrix}
\nu_t \\
\nu_{t-1}
\end{pmatrix} \nonumber \\
\tilde{y}_t &= I_{\tilde n} x_t  \  \ , \ \  t \in 2 \mathbb{Z} \ . 
\label{eq: hf ss trans_FD:p7} 
\end{align}   }
\noindent
Finally, we investigate whether there are 
(different) parameter vectors such that the autocovariance function of $(\tilde y_t)_{t \in 2 \mathbb Z}$, i.e. $\tilde \gamma$, is the same. (Recall that $\theta^I := 
\left( 
\alpha_f^{I},
\alpha_s^{I},
1,
\beta_{s}^{},
\sigma_{ff}^{I},
\sigma_{fs}^{I},
\sigma_{ss}^{I}
\right)
$ and
$\theta^{II} := 
\left( 
\alpha_f^{II},
\alpha_s^{II},
1,
\beta_{s}^{},
\sigma_{ff}^{II},
\sigma_{fs}^{II},
\sigma_{ss}^{II}
\right)
$ as defined in Section~\ref{subsect: mixed frequency data and generic ident}.)
This is the case if $\tilde{A}$, and $\tilde{B}\Sigma \tilde B'$ are equal though $\theta^{I} \not= \theta^{II}$.  
Consider $\tilde{A}$, for $\alpha_f=0$, from the definition of $M_[\alpha, \beta]$ in equation (\ref{eq: VECM hf:p5}), we can choose some $\alpha_s^{I} \neq \alpha_s^{II}$ solving the quadratic equation $2 \alpha_s + \beta_{s} \alpha_s^2 - [M_{\alpha,\beta}]_{2,1}=0$, in which case $\tilde{A}$ remains the same.\\
Next we consider $\tilde{B} \Sigma \tilde{B}^{\prime}$. To get equality with some $\theta^I$, $\theta^{II}$, the covariances of the lagged fast variable have to be equal, this demands for $\sigma_{ff}^{I} =\sigma_{ff}^{II}$. Second, to get equality of the covariances of the cointegrating terms it is sufficient to look at the 
covariances of $\Delta_N y_t$.
Direct calculations show that (with $\alpha_f=0$, 
$\sigma_{ff}^{I} = \sigma_{ff}^{II}$)
{\scriptsize
\begin{align}
&\left[ \tilde{B} \Sigma \tilde{B}^{\prime} 
\right]_{(2:3,2:3)} \nonumber \\
=&
\begin{pmatrix}
\sigma_{ff} & \alpha_s \sigma_{ff}
+ (1 + \alpha_s \beta_{s}) \sigma_{sf}
\\
\alpha_s \sigma_{ff}
+ (1 + \alpha_s \beta_{s}) \sigma_{sf}
& 
\alpha_s
\left( \alpha_s \sigma_{ff}
+ (1 + \alpha_s \beta_{s}) \sigma_{sf} \right) +
\left( \alpha_s \sigma_{fs}
+ (1 + \alpha_s \beta_{s}) \sigma_{ss} \right)
(1 + \alpha_s \beta_{s}) 
	\end{pmatrix} \ . 
 \label{eq: VECM hf:p11}
\end{align}  }
This yields $\sigma_{sf}^{II} = \sigma_{fs}^{II}
= \frac{\alpha_s^{I} \sigma_{ff} 
- \alpha_s^{II} +
(1 + \alpha_s^{I} \beta_{s}) \sigma_{sf}^{I}
}{  (1 + \alpha_s^{II} \beta_{s}) }
$, where
$(1+ \alpha_s^{II} \beta_{s}) \not=0$. Finally, choose $\sigma_{ss}^{II}$ 
such that the (2,2)-elements of 
(\ref{eq: VECM hf:p11}) for parameters $\theta^{I}$ and $\theta^{II}$ become equal, this requires
$(1+ \alpha_s^{I} \beta_{s}) \not=0$. Hence, we have obtained a pair $\theta^I$, $\theta^{II}$, $\theta^I \not= \theta^{II}$, where $\tilde{A}$ and $\tilde{B} \Sigma \tilde{B}^{\prime}$ are the same. This implies by Lyapunov equations that the model parameters cannot be identified from $\tilde \gamma$. 
In this example we observed the following: If $\alpha_f =0$, $\left( y_{t}^f \right)_{ t \in \mathbb{Z} }$ is a random walk not affected by $\left( y_{t}^{s}
\right)_{t \in \mathbb{Z}}
$. Therefore, the high frequency variable $y_{t}^{f}$ does not provide information about the parameter $\alpha_s$. By contrast, if $\alpha_f \not=0$ the fast variable 
$y_{t}^{f}$ provides sufficient information to identify $\alpha_s$. To see this, if $\alpha_f \neq 0$, $\alpha_f$ follows from the (4,1)-element of $\tilde{\mathcal{A}}$, while $\alpha_s$ can be uniquely retrieved from the first coordinate of 
$M_{\alpha,\beta}$, that is from the equality
$2 \alpha_f + \alpha_f^2 + \alpha_f \beta_{s} \alpha_s
= \left[ M_{\alpha,\beta} \right]_{(1,1)}$.

\section{Proof of Theorem \ref{thm: topo thm}}
\begin{proof}
%
%
%
%
%
1. ($G\cap \Theta_1$ is dense.)\\
Suppose that $\theta_{0} \in \Theta_1$ does not satisfy at least one of the identifiability conditions. Let $\varepsilon >0$, we show that there exists $\theta \in G \cap \Theta_1$ such that $\norm{\theta - \theta_0} < \varepsilon$ by perturbing the eigenvalues / eigenvectors of the companion matrix $A$ corresponding to $\theta_0$.\\
For this we define a mapping $f_{\theta_0}$ that maps $\mathcal{A}$ to a companion matrix $\mathcal{A}^*$ with perturbed eigenvalues and eigenvectors such that $\theta = \vect \begin{pmatrix}
 \mathcal{A}_1^*  & \cdots & \mathcal{A}_p^*	
 \end{pmatrix}
$ is in $G\cap \Theta_1$: 
\begin{enumerate}
	\item Compute the Jordan decomposition of $\mathcal{A} = Q \Lambda Q^{-1}$.
	\item Perturb the eigenvalues: 
\begin{align}
	\bar{\mathcal{A}}^* = Q \underbrace{ \big( \Lambda +\diag(\underbrace{0, ..., 0}_\text{$n-r$ -times}, \xi_{1}, ..., \xi_{np-(n-r)})\big)}_\text{$\tilde \Lambda$} Q^{-1}. \label{eq: perturbed eigenvals}
\end{align}

%
%
%
\item We transform $\bar{\mathcal{A}}^*$ to a similar matrix $\mathcal{A}^*$ that has the companion structure by using the procedure from \citet{AndersonEtAl2016a}:
\begin{align*}
\mathcal{A}^* = T \bar{\mathcal{A}}^* T^{-1},	\qquad \mbox{hence   } \quad 
\mathcal{A}^*  T = \begin{pmatrix}
\mathcal{A}_1^* &\cdots  & & \mathcal{A}_p^* \\
I_n   &        & &  0     \\ 
      & \ddots & & \vdots \\
      &        & I_n &	0
\end{pmatrix} 
\begin{pmatrix}
	T_1 \\
    T_2 \\
	\vdots \\
	T_p
\end{pmatrix}
 = T \bar{\mathcal{A}}^* \ , 
\end{align*}
where $T_j$ for $j = 1, ...,p$ are the $n \times np$ rowblocks of $T$. Now we set $T_1 = \begin{bmatrix} I_n & 0 & \cdots & 0 \end{bmatrix}$ and solve the equation above: 
\begin{align*}
\mathcal{A}^*_1 T_1 + ... + \mathcal{A}_p^* T_p = T_1 
\bar{\mathcal{A}}^*, \qquad
T_1 = T_2 \bar{\mathcal{A}}^*, \qquad 
    \dots, \qquad 
T_{p-1} = T_p \bar{\mathcal{A}}^*,
\end{align*}
which yields 
\begin{align*}
T_{j} = T_{j-1} \bar{\mathcal{A}}^{* -1} \quad \mbox{for } j = 2, ..., p.
\end{align*}  
\end{enumerate}
Clearly, the mapping $f_{\theta_0}: \xi \mapsto \mathcal{A}^* \mapsto \theta$ for $\xi = (\xi_{n-r+1}, ..., \xi_{np})' \in \mathbb R^{np - (n-r)}$ is continuous at  $\theta_0$ and $f_{\theta_0}(0) = \theta_0$ (as in this case $T = I_{np}$). So for the $\varepsilon$-neighborhood around $\theta_0$ denoted by $U_\varepsilon(\theta_0)$ there exists a $\delta > 0$, such that for all $\xi \in U_\delta (0)$ we have $f_{\theta_0}(\xi) \in U_\varepsilon(\theta_0)$, where $U_\delta(0)$ is the open $\delta$-neighborhood around $0$ in $\mathbb R^{np-(n-r)}$.\\
Now, $\lambda^*:= (1,...,1, \lambda_{n-r+1} + \xi_{1},..., \lambda_{np} + \xi_{np-n-r})$ are the eigenvalues of $\mathcal{A}^*$ because they are the zeros of the characteristic polynomial of $\tilde \Lambda$ in equation (\ref{eq: perturbed eigenvals}) from which we obtain $\mathcal A^*$ by similarity transformation with $TQ$. For any $\delta > 0$, we can find a $\xi \in U_\delta(0)$ such that the corresponding eigenvalues $\lambda^*$ of $ \mathcal{A}^*$ satisfy the conditions \textbf{(I1)}, \textbf{(I3)}, \textbf{(I4)} and \textbf{(I5)}.
Analogously to equation (\ref{eq: perturbed eigenvals}), we can perturb the eigenvalues and eigenvectors of $A$ to ensure conditions \textbf{(I5)} (second part) and \textbf{(I6)}.\\
We have to ensure that the image $f_{\theta_0}(\xi)$ is real valued: Since $\mathcal{A}$ is real valued, for any complex eigenvalue $z = a + ib \in \mathbb C  \setminus \mathbb R$, the conjugate $\bar z = a - ib$ is also an eigenvalue of $\mathcal{A}$. If the algebraic multiplicity of $z$ is larger than $1$, $z$ has to be perturbed. As is easily shown, if we add to $z$ and $\bar z$ the same small real number, the resulting $\bar{\mathcal{A}}^*$ (and therefore also $\mathcal{A}^*$) is again real valued.\\
Thus, we found $\theta \in G$ close to $\theta_0$ and are left with checking whether $\theta$ is also in $\Theta_1$. \textbf{(C3)} is trivial.\\
For \textbf{(C1)}, note that, still $n-r$ eigenvalues of $\mathcal{A}^*$ equal unity which ensures that $\rank \Pi = r$ \citep[see][]{bauerwagner2012}. Applying the procedures described above, we obtain the vector error correction representation corresponding to $f_{\theta_0}(\xi) = \theta$, say  $(\alpha(\xi), \beta(\xi), \Phi_1(\xi),..., \Phi_{p-1}(\xi))$, and see that
 \begin{align*}
 	g:\theta \mapsto \det \alpha_\bot(\xi) ' \big(I_n - \sum_{j = 1}^{p-1} \Phi_j(\xi)\big)\beta_\bot(\xi)
 \end{align*}
is continuous at $\theta = \theta_0$. We know that $g(\theta_0) \neq 0$ since $\theta_0 \in\Theta_1$. So there exists $\varepsilon_3 > 0$ such that the neighbourhood $U_{\varepsilon_3}\big( g(\theta_0)\big)$ is bounded away from zero. By continuity there exists $\varepsilon_2 > 0$ such that for all $\theta \in U_{\varepsilon_3}(\theta_0)$, we have $g(\theta) \in U_{\varepsilon_2}\big(g(\theta_0)\big)$. For the same reasons as above we can find suitable $\xi$ such that $f_{\theta_0}(\xi) = \theta \in U_{\varepsilon}(\theta_0) \cap U_{\varepsilon_3}(\theta_0)$. Hence $\Theta \cap G$ is dense in $\Theta$. \\
2. ($G\cap \Theta_1$ is open in $(\Theta_1, d)$), \\
where $d$ denotes the Euclidean metric.
Suppose now for $\theta^* \in G \cap \Theta_1$, we have to show that there exists $\varepsilon > 0$ such that $U_\varepsilon(\theta_0) \subset G \cap \Theta_1$. 
The eigenvalues are the zeros of the characteristic polynomial of $A$ and therefore continuous functions at $\theta^*$ (since as is well known, the zeros of any polynomial are continuous function of its coefficients).
So the mapping 
\begin{align*}
e: \theta \mapsto \mathcal{A} \mapsto \begin{pmatrix}
 \lambda_{n-r+1} & \cdots & \lambda_{np} 	
 \end{pmatrix} = \lambda	
\end{align*}
is continuous in $\theta^*$. Clearly there is an open neigbourhood $U\subset \mathbb C^{np - (n- r)}$ of $\lambda^* = e(\theta^*)$ such that for all $\lambda \in U$ the corresponding spectrum $\begin{pmatrix}
	1 & \cdots & 1 & \lambda'
\end{pmatrix}'$ satisfies the identifiability conditions. The pre-image $e^{-1}(U) \subset G$ is an open neighborhood of $\theta_0$. Analogously to the arguments applied above, we can establish \textbf{(C2)}.\\
\textbf{(I2)} follows from \textbf{(C4)} and \textbf{(I1)} which completes the proof.
\end{proof}
\section{Proof of Theorem \ref{thm: th1aggregated}}
\begin{proof}
    This follows from transforming a state-space system into prediction error form. See \cite{HannanAndDeistler2012}[chapter 1] and \cite{GersingAndDeistler2021}. From \cite{Johansen1995}[Proof of Theorem 4.2] it follows that the largest eigenvalue of $A$ is in modulus smaller than one. Hence the system is stable. The linear expansion of the transfer function for a stable system is already the Wold representation as the inputs 
    $\tilde \nu_t$ are the innovations. Hence, the system is also miniphase
    \citep[see, e.g.,][Chapters~2 and 7.3]{ScherrerDeistler2018book}.
\end{proof}
\section{Proof of Theorem \ref{thm: Minimal}}
\begin{proof}
By \cite{Johansen1995}[Proof of Theorem 4.2] it follows that the eigenvalues of modulus smaller than $1$ are the same, for $\mathcal{A}$ and 
${A}$.\\
\textit{1.1 Observability for the Stock Case:} We use the PBH-Test 
\citep[see, e.g.,][Section~2.4.3]{Kailath1980}
to prove that the pair $(A_b, C_b)$ is generically observable (note that the observability of 
$(A_b,  C_b)$ also implies the 
observability of $(A_{b,c}, C_{b,c})$ since $c$ is non-singular). For this, note that the eigenvectors of $A_b$ are the same as the eigenvectors of $A$. Let $\lambda$ be an eigenvalue of $A$ and $q = \begin{pmatrix}
q_\beta' & q_1' &  \cdots & q_{p-1}'
\end{pmatrix}'$ the corresponding eigenvector. We write 
\begin{align*}
 Aq=   \begin{bmatrix}
    \beta'\alpha + I_r & \beta' \Phi_1 & \cdots & & \beta'\Phi_{p - 1} \\
    \alpha & \Phi_1 & \cdots & & \Phi_{p - 1} \\
    0     & I_n     &     &  & 0    \\
    \vdots      &         & \ddots & &    \\
     0     &         &    & I_n & 0
    \end{bmatrix} 
    \begin{pmatrix}
    q_{\beta} \\
    q_1 \\
    \vdots \\
    \\
    q_{p - 1}
    \end{pmatrix}
    = \lambda \begin{pmatrix}
    q_{\beta} \\
    q_1 \\
    \vdots \\
    \\
    q_{p - 1}
    \end{pmatrix}, 
\end{align*}
where $q_\beta$ is $r \times 1$ and $q_i$ is $n \times 1$ for $i = 1, ... , p-1$. From this, we obtain the relations 
\begin{align}
    (\beta'\alpha + I_r) q_\beta + \sum_{i = 1}^{p - 1} \beta' \Phi_i q_i &= \lambda q_\beta  \label{eq: eigenvectors tilde A first rowblock}\\
    \alpha q_\beta + \sum_{i = 1}^{p - 1} \Phi_i q_i &= \lambda q_1 \label{eq: eigenvectors tilde A second rowblock} \\
    q_i &= \lambda q_{i + 1}  \ ,  \qquad i = 1, ..., p - 2.
\end{align}
Since $A$ is of full rank, $\lambda \neq 0$ and $q_1 = 0$ imply $q = 0$, which is a contradiction (noting that $\alpha$ has rank $r$). Now we look at 
\begin{align}
    C_b q = \begin{pmatrix}%
	I_r & 0 & \cdots &   &  &  & 0  \ \\
	0 & I_n & \cdots & I_n  & 0  & \cdots  & 0\\
	0 & 0 & \left( I_{n_f} , 0  \right) & \cdots  &  \left( I_{n_f} , 0  \right) & & \\
	& & & \ddots & & \ddots & \vdots \\
	&  &  & \left( I_{n_f} , 0  \right) & \cdots &  \left( I_{n_f} , 0  \right)  & 0 \\
	\end{pmatrix}  %
	\begin{pmatrix}
	\lambda^N q_\beta \\
	\lambda^N q_1 \\
	\vdots \\
	\lambda^N q_{p - 1}
	\end{pmatrix} \label{eq: tilde C_b times q} \ , 
\end{align}
which is not equal to zero. If, for example, 
\begin{align*}
    &\lambda^N q_1 + \cdots + \lambda^N q_{N}  = \lambda^N q_1 +  \lambda^{N - 1} q_1 + \cdots +  q_1  =(1 + \lambda + \cdots + \lambda^N)q_1 \neq 0 \\
    & \hspace{4cm} \Leftrightarrow (1 + \lambda + \cdots + \lambda^N) \neq 0,
\end{align*} 
which is generically the case (see Assumption~\ref{a: g-ident}).
Recall that by $q$ we denote eigenvectors of $A$ and by $v$ eigenvectors of $\mathcal{A}$, where both correspond to the same eigenvalue $|\lambda| < 1$. In Lemma \ref{lemma: q_beta and q_1 and v_1}, we show that
\begin{align}
q_\beta = \frac{\lambda}{\lambda - 1}\beta' q_1 = \beta'v_1, \label{eq: q_beta and q_1 and v_1}
\end{align}
so if we suppose that $v_1$ is not in the right kernel of $\beta'$, we also get $C_b q \neq 0$.\\
\\
\textit{1.2 Observability for the Flow Case:} 
The first part of the proof is analogous to the stock case. It remains to show that there exists no eigenvector that is in the right kernel of $C_b$, where $C_b$ is now defined in (\ref{eq: mixedf ss obs_FD7}). 
Now, analogously to the procedure in (\ref{eq: tilde C_b times q}) we obtain, that an eigenvector of $A_b$ is not in the rightkernel of $C_b$ if e.g. 
\begin{align*}
    & \lambda^N q_1 + \cdots + \lambda^N q_{N} - \lambda^N q_{N + 1} - \cdots - \lambda^{N}q_{2N}\\ &= \lambda^N q_1 +  \lambda^{N - 1} q_1 + \cdots +  \lambda q_1 - q_1  - \cdots - \lambda^{- N +1} q_1 \\
    &=\lambda^{N - 1}(- 1 - \lambda -  \cdots - \lambda^{N - 1} + \lambda^{N} + \cdots + \lambda^{2N - 1})q_1 \neq 0 \\
    &\Leftrightarrow (- 1 - \lambda -  \cdots - \lambda^{N - 1} + \lambda^{N} + \cdots + \lambda^{2N - 1})\neq 0,  
\end{align*}
Also the second part is similar to the stock case: By Lemma \ref{lemma: q_beta and q_1 and v_1}, 
$q_\beta = \frac{\lambda}{\lambda - 1}\beta' q_1 = \beta'v_1$. Assume that $v_1$ is not in the right kernel of $\beta'$ (as already done in the stock case). In addition, by considering the first $r$ rows of the matrix $c$ for the flow case, provided in (\ref{eq:cmatrix_FD_flow}), we get
{\small
\begin{align*}
    & 
N I_r q_{\beta} -(N-1) \beta^{\prime} q_1  -(N-2) \beta^{\prime} q_2 - \cdots -2 \beta^{\prime} q_{N-2}  -\beta^{\prime} q_{N-1} \\
=& 
N I_r \frac{\lambda}{\lambda - 1}\beta' q_1 - \frac{(N-1)}{\lambda^0} \beta^{\prime} q_1  - \frac{(N-2)}{\lambda} \beta^{\prime} q_1 - \cdots -
- \frac{2}{\lambda^{N-2}} \beta^{\prime} q_{1}
- \frac{1}{\lambda^{N-1}} \beta^{\prime} q_{1} \nonumber \\
=& 
\left( N \frac{\lambda}{\lambda - 1}  - \frac{(N-1)}{\lambda^0}   - \frac{(N-2)}{\lambda}  - \cdots 
- \frac{2}{\lambda^{N-2}}
- \frac{1}{\lambda^{N-1}} \right)
\beta^{\prime} q_1
\nonumber \\
=& 
\frac{1}{\lambda^{N-1}}
\left( N \frac{\lambda^N}{\lambda-1}  - (N-1) \lambda^{N-1}   - (N-2) \lambda^{N-2}  - \cdots - 2 \lambda - 1 \right)
\beta^{\prime} q_1 \ . 
\end{align*} }
Note that $\lambda \not=1$ and $\lambda \not=0$ by the model assumptions (recall that by \cite{Johansen1995}[Proof of Theorem 4.2] it follows that the eigenvalues of modolus smaller than $1$ are the same, for $\mathcal{A}$ and 
${A}$). Hence, if 
$v_1$ is not in the right kernel of $\beta'$ and
$N \frac{\lambda^N}{\lambda-1}  - (N-1) \lambda^{N-1}   - (N-2) \lambda^{N-2}  - \cdots - 2 \lambda - 1 \not=0$ we also get  that $\tilde C_b q \neq 0$ for the flow case.\\
\textit{2. Controllability:} 
It is enough to show that the matrix $\E x_{t + 1} \begin{pmatrix}
\tilde y_t' & \tilde y_{t - N}' & \cdots
\end{pmatrix}'$ has full rank. For $k$ sufficiently large, we have
\begin{align*}
    x_{t + N - 1} & = A_{b, c}^{k - 1} x_{t - kN + 1} + \sum_{j  = 0}^{k - 1} A_{b, c}^j B_{b, c} \nu^b_{t - N - j N} \\
    \Delta_N y_{t - k N} & = \underbrace{\begin{bmatrix}
    0_{n \times r} & I_n & 0 & \cdots & 0
    \end{bmatrix}}_{S_{\Delta_N y}} x_{t - kN + 1} \\
    \E \Delta_N y_{t - kN} x_{t - N + 1}' &= \E \bigg \{ S_{\Delta_N y} x_{t - kN + 1} x_{t - kN + 1}' {A_{b, c}^{k - 1}}' + S_{\Delta_N y} x_{t - kN + 1} \bigg( \sum_{j = 0}^{k - 2} A_{b, c}^{j} B_{b, c} \nu^b_{t - N - jN} \bigg)' \bigg\} \\
    &= S_{\Delta_N y} \underbrace{c \Gamma_{rp} c'}_{\Gamma_{rp, c}} {A_{b, c}^{k - 1}}'. 
\end{align*}
Therefore 
\begin{align*}
    \E x_{t - N + 1} & \begin{pmatrix}
    \Delta_N y_{t - kN} ' & 
    \Delta_N y_{t - (k + 1)N} ' &
    \cdots &   
    \Delta_N y_{t - (k + p - 1)N} '
    \end{pmatrix}  \\
    &=  A_{b, c}^{k - 1} \begin{bmatrix}
      \Gamma_{rp, c} S_{\Delta_N y}' & 
    A_{b, c} \Gamma_{rp, c} S_{\Delta_N y} ' &
    \cdots &
     A_{b,c} ^{p - 1}  \Gamma_{rp, c} S_{\Delta_N y}'
    \end{bmatrix}, 
\end{align*} 
which has full rank if $\Gamma_{rp} > 0$ as follows from the proof of Theorem 7 in \cite{AndersonEtAl2016a}.
%
By \citet{HannanAndDeistler2012}[Theorem~2.3.3] controllability and observability imply that the system is minimal.
\end{proof}
\begin{lemma}\label{lemma: q_beta and q_1 and v_1}
Suppose the \textcolor{red}{Assumption \ref{ass: Theta}} and \ref{a: g-ident} hold. Then equation (\ref{eq: q_beta and q_1 and v_1}) holds.
\end{lemma}
\begin{proof}
Substracting $\beta'$ times (\ref{eq: eigenvectors tilde A second rowblock}) from (\ref{eq: eigenvectors tilde A first rowblock}), we obtain 
\begin{align*}
q_{\beta}  = \lambda  q_\beta - \lambda \beta^{\prime} q_{1}    \ \text{such that} \ 
 q_{\beta}  =  \frac{\lambda}{\lambda-1} \beta^{\prime} q_{1}.
\end{align*}
Next, we consider the eigenvector $v = \begin{pmatrix}
v_1' & \cdots & v_p'
\end{pmatrix}'$ of $\mathcal{A}$ corresponding to $\lambda$ (recall that eigenvalues in modulus smaller that one of $A$ and $\mathcal{A}$ are the same). By using the relations of the parameters between the VECM and VAR representation, we get 
\begin{align*}
\lambda v_1 &= ( I_n + \alpha \beta^{\prime} ) v_1
 + \Phi_1 (v_1- v_2) + \Phi_2 (v_2 -  v_3) + 
\dots
+ \Phi_{p-1} (v_{p-1} - v_p) \\
	& \alpha \beta^{\prime} v_1
	+  \Phi_1 \frac{\lambda-1}{\lambda} v_1 + \Phi_2 \frac{\lambda-1}{\lambda^2} v_1  + 
	\dots
	 \frac{\lambda-1}{\lambda^{p-1}} v_1 = (\lambda-1) v_1  \ ,
\end{align*}
where the last relation follows from $v_i = \lambda v_{i + 1}$ for $i = 1, ..., p - 1$, which results by the companion structure of $A$. Now, we see that $q_1 = ((\lambda - 1) / \lambda) v_1$ solves (\ref{eq: eigenvectors tilde A first rowblock}) and (\ref{eq: eigenvectors tilde A second rowblock}).  
\end{proof}
\section{Proof of Theorem \ref{thm: identifiability in steps}}
\begin{proof}
Consider the stable, miniphase spectral factor $\tilde k(\tilde z)$, $\tilde z := z^N$, corresponding to the Wold representation of $(\tilde y_t)_{t \in N \mathbb Z}$.\\
\textbf{Step 1:} We obtain an arbitrary minimal realisation $(\bar A_{b, c}, \bar B_{b, c}, \bar C_{b, c})$ of $\tilde k(\tilde z)$, e.g. by taking the echelon form, see \cite{HannanAndDeistler2012}[Thm 2.5.2.].\\
{\bf Step 2:} (Obtain eigenvalues $\Lambda
	= diag \left(\lambda_1,\dots,\lambda_{r+n(p-1)} \right)
	$ and a linear combination of the eigenvectors of ${{A}}_{}$, denoted $q_i$,
	from $\bar{{A}}_{b,c}$). \\
By, e.g., \citet{HannanAndDeistler2012}[Theorem~2.3.4] the parameter matrices of  minimal systems relate via $\bar{A}_{b,c} = T^{-1} A_{b,c} T^{} $, $\bar{C}_{b,c} = C_{b,c} T^{} $ and $\bar{B}_{b,c} = T^{-1} \tilde{B}_{c}$, where $T$ is a non-singular matrix.\\ 
Since ${\mathcal A}$ 
(see equation (\ref{eq: hf ss trans}))
is assumed to be diagonalizable (Assumption \ref{a: g-ident}), the matrix ${A}$  (see 
equation~(\ref{eq: hf ss trans_FD}); 
recall that by \cite{Johansen1995}[Proof of Theorem 4.2] the  of modulus smaller than $1$ are the same, for $\mathcal{A}$ and 
${A}$) can be expressed by means of ${A} = Q \Lambda Q^{-1}$, where $\Lambda = diag(\lambda_1,\dots,\lambda_{r+n(p-1)})$ is the diagonal matrix of eigenvalues of ${A}$ and $Q = \left( q_1,\dots,q_{r+n(p-1)} \right)$ contains the eigenvectors. ${A}_{b} = A^N$, $C_{b,c} =  C_b c^{-1}$ and  ${A}_{b,c} = c {A}_b c^{-1}$, such that $\bar{A}_{b,c} = T^{-1} {A}_{b,c} T^{} = T^{-1} 
	c A_{b} c^{-1} T^{}
= \left(T^{-1} c Q \right) \Lambda^N \left(T^{-1} c Q \right)^{-1}$, $\bar{C}_{b,c} =  {C}_{b,c} T^{} ={C}_{b} c^{-1} T^{} $. 
By the eigen-decomposition of $\bar{A}_{b,c}$, we obtain $\left(T^{-1} c Q \right)$ and $\Lambda^N$. 
In addition, $\left({0}_{n \times r} , {0}_{n \times n} , {I}_n \ {0} \dots {0} \right) {A}^2
	= \left({0}_{n \times r} ,{I}_n \ {0} \dots {0} \right) {A}$ by the companion structure of ${A}$. Hence by (\ref{eq: C_b structural matrix definition}), we have  
{\footnotesize
\begin{align}
\bar{{C}}_{b,c} T^{-1} c Q =  C_b c^{-1}T T^{-1}c Q = C_b Q =
	\begin{pmatrix}
\begin{pmatrix}I_r & 0 & \cdots & 0 \end{pmatrix} A^N\\
  S_{n_f}^{(1)}      A^N\\
 S_{n_s}^{(1)}       A^N\\
 S_{n_f}^{(1)}  A^{N - 1}\\
   S_{n_f}^{(1)} A^{N - 2}\\
\vdots  \\
S_{n_f}^{(1)}  A  \\
\end{pmatrix}Q.
	\end{align} }
Now we look at the last two rowblocks of $C_b$ with the eigenvectors $q_i$, $1 \leq i \leq m$. From assumption \ref{a: g-ident} \textbf{(I4)}, it follows that the eigenvectors of $A$ are the same as the eigenvectors of $A^2$ (also as $A^N$) \citep[see ][Lemma 3.2.1]{Felsenstein2014}, therefore we have  
\begin{align}
    S_{n_f}^{(1)} A^2 q_i = S_{n_f}^{(1)} \lambda_i^2 q_i \nonumber \\
    S_{n_f}^{(1)} A q_i = S_{n_f}^{(1)}  \lambda_i q_i, \label{eq: S_nf A q undsoweiter}
\end{align}
and we can compute all eigenvalues not equal to one since $S_{n_f}^{(1)} q_i \neq 0$ by assumption \ref{a: g-ident} \textbf{(I6)}. The flow-case is analogous. 
Summing up, from 
$\bar{{A}}_{b,c}$ we are able to obtain $T^{-1} c Q $, $\Lambda^N = diag(\lambda_1^N,\dots,\lambda_{r+n(p-1)}^N)$ and $\Lambda = diag(\lambda_1,\dots,\lambda_{r+ n(p-1)})$.\\
\\
{\bf Step 3:} (relate $c^{-1} T$ to $T$) \\
To jointly treat the stock and the flow case, we write 
{
	\begin{align*}
c = \begin{pmatrix}
	c_{\beta \beta} & c_{\beta 1}  & c_{\beta 2}  & \cdots & 
	c_{\beta N-1} & 0 & \cdots
	&  \\
	0   & c_{11} & c_{12} & \cdots & \cdots &c_{1 N} & c_{1, N+1}  & \cdots & \\ 
	\vdots & &  & \ddots & &  & & \ddots & &  \\
	\end{pmatrix}  \ ,	\end{align*}  }
where $c_{\beta 1},\dots,c_{\beta N-1}$ and
$c_{N-1+j}$, $j \geq 1$
are zero for the stock case (see equation (\ref{eq:cmatrix_FD})). For the flow case $c_{11},\dots,c_{N-1,1} = I_n$
and
$c_{1N},\dots,c_{2N-1,1} = -I_n$ (see equation (\ref{eq:cmatrix_FD_flow})). For the case of stock and flow variables 
the corresponding coordinates of 
$c_{1N},\dots,c_{2N-1,1}$ are zero for stock variables.\\[0.5em]
	Let  
	{
	\begin{align}
		\label{eq:RandD}
		A =	\left(  
		\begin{array}{c} 
		A_{\beta} \\
		A_{1} \\
		\vdots \\
		A_{p-1} \\
		\end{array}
		\right) \ , \ \
		T =	\left(  
		\begin{array}{c} 
		T_{\beta} \\
		T_{1} \\
		\vdots \\
		T_{p-1} \\
		\end{array}
		\right) \ , \text{ and } \
		R &:= c^{-1} T  =
		\left(  
		\begin{array}{c} 
		R_{\beta} \\
		R_{1} \\
		\vdots \\
		R_{p-1} \\
		\end{array}
		\right)
		\ .
		\end{align} 
	}
	\noindent Observe that for the stock case (the flowcase is treated analogously) 
	{\scriptsize
		\begin{eqnarray}
		\bar{{C}}_{b,c} \ \bar{{A}}_{b,c}^{-1} &=& 
		\left( 
		\left(  
		\begin{array}{ccccccc} 
		{I}_{r} & c_{\beta 1} & \dots &  c_{\beta N-1} & 0 & \dots & \\
		{0}  & I_n  & I_n  & \dots &I_n & 0 & \dots   \\
		& & & \sharp &  & & 
		\end{array}
		\right) A^N \right) c^{-1} T
		\underbrace{T^{-1} c A^{-N} c^{-1} T}_{\bar{A}_{b,c}^{-1} }
		\nonumber  \\
		&=& 
		\left(  
		\begin{array}{ccccccc} 
		{I}_{r} & c_{\beta 1} & \dots &  c_{\beta N-1} & 0 & \dots & \\
		{0}  & I_n  & I_n  & \dots &I_n & 0 & \dots   \\
		& & & \sharp &  & & 
		\end{array}
		\right)
		c^{-1} T 
		= 
		\left( 
		\begin{array}{c} 
		\left[ c \right]_{(1:n+r,1:m)} \\
		\sharp 
		\end{array}
		\right) c^{-1} T 
		=
		\left( 
		\begin{array}{c}
		T_{\beta} \\ 
		T_1 \\
		\sharp
		\end{array}
		\right) \ \ \ \ \ \ \ \  
		\label{eq: appa_FD T_beta T_1}
		\end{eqnarray} }
	where ``$\sharp$'' denotes some matrix entries which are not important here.
	Note that $\bar{A}_{b,c} = T^{-1} c A_b c^{-1} T^{}$. From Steps 1 and 2, we obtain $\bar{A}_c := T^{-1} c A c^{-1} T^{} = T^{-1} c Q \Lambda Q^{-1}c^{-1} T$. $A c^{-1} T =  c^{-1} T \bar{A}_c$ and
	{\scriptsize
		\begin{align}
		& AR =  A \begin{pmatrix}
		  R_\beta \\
		  R_1 \\
		  \vdots \\
		  R_{p-1}
		\end{pmatrix}  = 
		\left(
		\begin{array}{c}
		(I_r + \beta' \alpha  ) R_{\beta} + \beta' \Phi_1 R_1 + \dots +\beta' \Phi_{p-1} R_{p-1}   \\
		\alpha  R_{\beta} +  \Phi_{ 1} R_1 + \dots + \Phi_{p-1} R_{p-1}   \\
		R_1 \\
		\vdots \\
		R_{p-2}    \\
		\end{array}
		\right) = 
		 \begin{pmatrix}
		  R_\beta \bar A_c\\
		  R_1 \bar A_c\\
		  R_2 \bar A_c \\
		  \vdots \\
		  R_{p-1} \bar A_c
		\end{pmatrix}  = R \bar A_c\\[0.5em]
		\label{eq:7.15appa2_FD}
		%
		%
		& AR = R \bar A_c
		= c^{-1}
		\left(
		\begin{array}{c}
		T_{\beta}  \bar{A}_{c}  \\
		T_1  \bar{A}_{c}  \\
		T_2  \bar{A}_{c}  \\
		\vdots \\ 
		T_p  \bar{A}_{c}  \\
		\end{array} 
		\right)  
		=
		c^{-1}\left(
		\begin{array}{c} 
		T_{\beta}  T^{-1} c A c^{-1} T  \\
		T_1  T^{-1} c A c^{-1} T  \\
		T_2  T^{-1} c A c^{-1} T  \\
		\vdots \\ 
		T_p  T^{-1} c A c^{-1} T   \\
		\end{array} 
		\right) 
		\; . 
		\end{align} 
	}
\noindent Now, $R_{\beta} = {c}_{\beta}^{-1} T$ and $R_1 = {c}_1^{-1} T$, where
${c}_{\beta}^{-1}  := \left[ c^{-1} \right]_{(1:r,1:m)}$ and ${c}_1^{-1}  := \left[ c^{-1} \right]_{(r+1:r+n,1:m)}$. Therefore, we receive $R_i$ for $i = 2, ..., p-1$, given $R_1 = {c}_1^{-1} T_1$ from the recursion $R_{i+1} = R_i \bar{A}^{-1}_c$, for $i = 1, ..., p-2$.
%
%
\\
{\bf Step 4:} (obtain $R=c^{-1} T$, $T$ and $\beta, \Phi_1 , ..., \Phi_{p-1}$ ) \\	
To retrieve $T$ and $R$ we proceed as follows: By means of (\ref{eq: appa_FD T_beta T_1}) and (\ref{eq:7.15appa2_FD}), and the assumption $p \geq 2 N$ we derive
{\scriptsize
\begin{eqnarray}
\label{eq:agg3_FD}
&&  T = c R 
=
\left( 
\begin{array}{l}
c_{\beta \beta} R_{\beta} + c_{\beta 1} R_1 +  c_{\beta 2} R_2 + \dots  + c_{\beta N-1} R_{N-1} \\	
0  R_{\beta} + c_{11} R_1 +  c_{12} R_2 + \dots  + c_{1N} R_{N}
+
c_{1,N+1} R_{N+1}  + \dots  + c_{1,2N} R_{2N}
\\
	0 R_{\beta} + 0 R_{1} + R_2 +  R_3 + \dots + R_{N+1} \\
	\vdots \\
	R_N +  R_{N+1} + \dots  + R_{2N-1}  \\
	\hline 
	I_n R_{N+1} \\
	\vdots \\
	I_n R_{p-1} 
	\end{array} \right) \; .   \ \ \
	\end{eqnarray} }
	%
%
%
Recall that for stock case $c_{1j}= I_n$, $j=1,\dots,N$, $c_{1j}= 0$, $j>N$, $c_{\beta j}=0$, for $j \geq 1$, while for the flow case
$c_{1j}= I_n$, $j=1,\dots,N$, $c_{1j}= -I_n $, $j=N+1,\dots,2N$, and $c_{\beta j} 
= - (N-j) \beta^{\prime}$, for $j = 1,\dots,N-1$.

From the above considerations $T_1$ can be obtained from (\ref{eq: appa_FD T_beta T_1}). Since $R_{i+1} = R_{i} \bar{A}_c^{-1}$, equation (\ref{eq:agg3_FD}) yields
{\footnotesize
\begin{eqnarray}
	T_1 &=&
	\begin{cases}
	R_1 + R_2 + \dots + R_{N} &, \text{ for the stock case}, \\
	R_1 + R_2 + \dots + R_{N} 
	-
	R_{N+1}  - \dots - R_{2N} 
	&, \text{ for the flow case}.
	\end{cases}
\end{eqnarray} }
In the above Step 3, we obtained $R_{i+1} = R_i \bar{A}^{-1}_c$, which results in
{\footnotesize
\begin{eqnarray}
\label{eq:agg3A_FD}
	T_1 &=&
	\begin{cases}
	R_1 + R_2 + \dots + R_{N} &, \text{ for the stock case}, \\
	R_1 + R_2 + \dots + R_{N} 
	-
	\left( R_{1}  + \dots + R_{N} \right) \bar{A}^{-N}_c 
	&, \text{ for the flow case},
	\end{cases} 
\end{eqnarray} }
such that 
$R_{1}  + \dots + R_{N}	= T_1$ for the stock and
$R_{1}  + \dots + R_{N}	= T_1 \left( I_m - \bar{A}^{-N}_c \right)^{-1} $ for the flow case. 
As already obtained above,  $R_{i+1} = R_i \bar{A}^{-1}_c$. This yields 
$R_{1}  + \dots + R_{N}
= R_1 \sum_{j=1}^N \bar{A}^{-j+1}_c $. Since 
$R_{1}  + \dots + R_{N}$ follows from 
(\ref{eq:agg3A_FD}) we are also able to derive $R_1$ and therefore $R_{i+1}$ by the recursion
$R_{i+1} = R_i \bar{A}^{-1}_c$, $i=2,\dots,p-1$.
Finally, we observe
{\small
\begin{eqnarray}
\label{eq:agg3A_FD:2}
T_2 &=&
	R_2 + R_3 + \dots + R_{N+1} =
	\left(R_1 + R_2 + \dots + R_{N} \right) 
	\bar{A}_c^{-1}  \nonumber \\
	&& \vdots \nonumber \\
	T_{N} &=&
	R_{N} + R_{N+1} + \dots + R_{N+N-1} =
	\dots  \nonumber \\
	T_{N + 1 } &=& R_{N + 1 }  \nonumber \\
	&& \vdots \nonumber \\
	T_{p-1} &=& R_{p-1}   \ .
	\end{eqnarray} }
	Hence $T_{i}$, $i=2,\dots,p-1$, are provided by (\ref{eq:agg3A_FD}). Recall that $T_{\beta}$ and $T_1$ follow from (\ref{eq: appa_FD T_beta T_1}). 

{\bf Step 5:}  (Obtain $\Sigma_{{\nu}}$)
	\\
Let  
{\footnotesize
\begin{align}
    \gamma_{\Delta_N y} (\kappa-\ell) & := \E \Delta_N y_{ t - \ell} \Delta_N y_{t-\kappa}'
    \ , \nonumber \\
    \gamma_{\beta} (\kappa - \ell) &:= \E \beta^{\prime} y_{t-\ell} (\beta^{\prime} y_{t-\kappa})'  \ ,  \\
    \gamma_{\beta, \Delta_N y} (\kappa - \ell) &:= \E \beta^{\prime} y_{t-\ell} \Delta_N y_{t-\kappa}'
    = \left( \E \Delta_N y_{t-\kappa} (\beta^{\prime} y_{t-\ell})^{\prime} \right)^{\prime }
    = \gamma_{\Delta_N y,\beta} (\ell - \kappa)^{\prime}
    \ , \ \text{ and } \nonumber \\
    \Gamma_{rp}  & := \E \underline{x}_{t+1} \underline{x}_{t+1}^{\prime} \in \mathbb{R}^{m \times m} \nonumber \\
    &=
    \left( 
    \begin{array}{ccccc}
    \gamma_{\beta} (0) & \gamma_{\beta, \Delta y} (0) &
    \gamma_{\beta, \Delta y} (1) & \dots &
    \gamma_{\beta, \Delta y} (p-2) \\
    \gamma_{\Delta y, \beta } (0) & \gamma_{\Delta y} (0) &
    \gamma_{\Delta y} (1) & \dots &
    \gamma_{\Delta y} (p-2) \\
    \gamma_{\Delta y, \beta } (-1) & \gamma_{\Delta y} (-1) &
    \gamma_{\Delta y} (0) & \dots &
    \gamma_{\Delta y} (p-3) \\
    & &  & \ddots & \\
    \gamma_{ \Delta y , \beta } (-p+2) & \gamma_{\Delta y} (-p+2) &
    \gamma_{\Delta y} (-p+3) & \dots &
    \gamma_{\Delta y} (0) 
    \end{array}
    \right)
    \nonumber \\
    &=
    \left( 
    \begin{array}{ccccc}
    \gamma_{\beta} (0) & \gamma_{\beta, \Delta y} (0) &
    \gamma_{\Delta y, \beta} (-1)^{\prime} & \dots &
    \gamma_{\Delta y,\beta} (p-2)^{\prime} \\
    \gamma_{\Delta y, \beta } (0) & \gamma_{\Delta y} (0) &
    \gamma_{\Delta y} (-1)^{\prime} & \dots &
    \gamma_{\Delta y} (-p+2)^{\prime} \\
    \gamma_{\Delta y, \beta }(-1) & \gamma_{\Delta y} (-1) &
    \gamma_{\Delta y} (0) & \dots &
    \gamma_{\Delta y} (-p+3)^{\prime} \\
    & & & \ddots & \\
    \gamma_{\Delta y, \beta } (-p+2) & \gamma_{\Delta y} (-p+2) &
    \gamma_{\Delta y} (-p+3) & \dots &
    \gamma_{\Delta y} (0) 
    \end{array}
    \right)
    \ , \label{eq:_sec_moms_forStep5} 
\end{align}
}
where $\underline{x}_t$ was defined in (\ref{eq: hf ss trans_FD}),(\ref{eq: hf ss obs_FD}). The last step follows from the fact that $\left(\underline{x}_t \right)_{t \in \mathbb{Z}}$ is stationary, such that $\Gamma_{rp}$ has to be symmetric.\\
Let $S_{\beta} := \left(I_{r \times r},  ,{0}_{r \times n}, \dots , {0} \right) \in \mathbb{R}^{r \times m}$, and
	$S_{\Delta_N y} := \left(0_{n \times r},  {I}_ n,{0}, \dots , {0} \right) \in \mathbb{R}^{n \times m}$. 
Then (\ref{eq: hf ss trans_FD}) and (\ref{eq: hf ss obs_FD}) result in	%
\begin{align*}
&	\gamma_{ u^{\cdot }} (-hN) := 
\E u_{t - h N}^{\cdot  } u_t^{\cdot  \prime}
=
S_{\beta } c {A}^{hN}
	c^{-1} c
	\Gamma_{rp} c^{\prime} S_{\beta}^{\prime} 
	= S_{\beta } {A}_{b,c}^h 
	c
	\Gamma_{rp } c^{\prime} S_{\beta}^{\prime} 
	\ , \nonumber  \\
&	\gamma_{u^{\cdot} , \Delta_N y} (-hN) 
:= 
\E u_{t - h N}^{\cdot  }   \Delta_N y_t^{\prime}
=  S_{\beta } c {A}^{hN}
	c^{-1} c
	\Gamma_{rp} c^{\prime} S_{\Delta_N y}^{\prime} 
	= S_{\beta } {A}_{b,c}^h 
	c
	\Gamma_{rp } c^{\prime} S_{\Delta_N y}^{\prime} 
	\ , \nonumber  \\
&	\gamma_{\Delta_N y} ( - hN ) = S_{\Delta_N y} c {A}^{hN} c^{-1} c 
	\Gamma_{rp } c^{\prime} S_{\Delta_N y}^{\prime}
	= S_{\Delta_N y } {A}_{b,c}^{h} 
	c
	\Gamma_{rp } c^{\prime} S_{\Delta_N y}^{\prime} 
	\ , \ \text{ and } \nonumber \\ 
&	\underbrace{\left(
	\begin{array}{cc}
	\gamma_{u^{\cdot} } \left( 0 \right)
	 & \gamma_{u^{\cdot},  \Delta_N y} \left( 0 \right)
	\\
	\gamma_{u^{\cdot}, \Delta_N y} \left( 0 \right) 
	&
	\gamma_{\Delta_N y} \left( 0 \right) 
	\\
	\gamma_{u^{\cdot}, \Delta_N y} \left( N  \right) &
	\gamma_{\Delta_N y} \left( N  \right) \\
	\vdots \\
	\gamma_{u^{\cdot}, \Delta_N y} \left( (np-2)N  \right) &
	\gamma_{\Delta_N y} \left( (np-2)N \right)
	\end{array}{}
	\right)}_{\Gamma_{\beta \Delta_N y }} 
	=
	\underbrace{
		\left(
		\begin{array}{c}
		S_\beta \\
		S_{\Delta_N y} \\
		S_{\Delta_N y} {A}_{b,c}^N \\
		S_{\Delta_N y} {A}_{b,c}^{2N} \\
		\vdots \\
		S_{\Delta_N y} {A}_{b,c}^{N (np-2)}
		\end{array}
		\right) }_{\mathcal {O}_N }
	c \Gamma_{rp} c^{\prime} 
	\left( 
	\begin{array}{c}
	S_{\beta}
	^{\prime} \\
	S_{\Delta_N y}
	^{\prime}
	\end{array} 
	\right).
\end{align*}
\noindent Note that $\mathcal {O}_N  {A}_{b,c}^{- N}  = \mathcal{O}$, where $\mathcal{O}$ is defined in (\ref{eq:def:matrixO}). The matrix $\mathcal {O}$ has full column rank, as will be shown in Lemma~\ref{lemma:CandO_FD}, such that also $\mathcal {O}_{N}$ has full rank. Thus we obtain the first two column blocks of $\Gamma_{rp, c}$.
Now looking at the specific structure of 
{\tiny
\begin{align}
\label{matrix:xx}
\Gamma_{rp, c} &=  
\left(
	\begin{array}{cc|ccccc|}
	\gamma_{u^{\cdot} } \left( 0 \right)
	 & \gamma_{u^{\cdot},  \Delta_N y} \left( 0 \right)
	& \gamma_{u^{\cdot},  \Delta_N y} \left( 1 \right)
	& \gamma_{u^{\cdot},  \Delta_N y} \left( 2 \right)
	& \cdots &
	& \gamma_{u^{\cdot},  \Delta_N y} \left( N-1 \right) \\
	\gamma_{ \Delta_N y , u^{\cdot}  } \left( 0 \right)
	 & \gamma_{\Delta_N y} \left( 0 \right)
	& \gamma_{\Delta_N y} \left( 1 \right)
	& \gamma_{\Delta_N y} \left( 2 \right)
	& \cdots &
	& \gamma_{\Delta_N y} \left( N-1 \right)
	\\
	\gamma_{ \Delta_N y , u^{\cdot}  } \left( -1 \right)
	 & \gamma_{\Delta_N y} \left( -1 \right)
	& \gamma_{\Delta_N y} \left( 0 \right)
	& \gamma_{\Delta_N y} \left( 1 \right)
	& \cdots &
	& \gamma_{\Delta_N y} \left( N-2 \right)
	\\
	\gamma_{ \Delta_N y , u^{\cdot}  } \left( -2 \right)
	 & \gamma_{\Delta_N y} \left( -2 \right)
	& \gamma_{\Delta_N y} \left( -1 \right)
	& \gamma_{\Delta_N y} \left( 1 \right)
	& \cdots &
	& \gamma_{\Delta_N y} \left( N-3 \right)
\\ 
\vdots &\vdots &\vdots &\vdots &\vdots &\vdots& \vdots
\\
%
%
	\gamma_{ \Delta_N y , u^{\cdot}  } \left( -(N-1) \right)
	 & \gamma_{\Delta_N y} \left( -(N-1) \right)
	& \gamma_{\Delta_N y} \left( -(N-2) \right)
	& \gamma_{\Delta_N y} \left( -(N-3) \right)
	& \cdots &
	& \gamma_{\Delta_N y} \left( 0 \right)
\\
\gamma_{\Delta y , u^{\cdot}  } \left( -N \right)
	 & \gamma_{\Delta y,  \Delta_N y} \left( -N \right)
	& \gamma_{\Delta y,  \Delta_N y} \left( -(N-1) \right)
	& \gamma_{\Delta y, \Delta_N y} \left( -(N-2) \right)
	& \cdots &
	& \gamma_{\Delta y ,  \Delta_N y} \left( -1 \right)
\end{array} 
	\right. \nonumber \\[2em]
& \hspace{2.5cm} \left.
	\begin{array}{|cccc}
 \gamma_{u^{\cdot},  \Delta y} \left( N \right) 
	&  \gamma_{u^{\cdot},  \Delta y} \left( N+1 \right)
	& \cdots 
	& \gamma_{u^{\cdot},  \Delta_N y} \left( p-2 \right)
\\
\gamma_{\Delta_N y , \Delta y} \left( N \right) 
	&  \gamma_{\Delta_N y, \Delta y} \left( N+1 \right)
	& \cdots 
	& \gamma_{\Delta_N y, \Delta_N y} \left( p-2 \right)
	\\
 \gamma_{\Delta_N y , \Delta y} \left( N-1 \right) 
	&  \gamma_{\Delta_N y, \Delta y} \left( N \right)
	& \cdots 
	& \gamma_{\Delta_N y, \Delta_N y} \left( p-3 \right)
	\\
\gamma_{\Delta_N y , \Delta y} \left( N-2 \right) 
	&  \gamma_{\Delta_N y, \Delta y} \left( N-1 \right)
	& \cdots 
	& \gamma_{\Delta_N y, \Delta_N y} \left( p-4 \right)
\\ 
 \vdots&\vdots&\vdots&\vdots 
\\
%
%
 \gamma_{\Delta_N y , \Delta y} \left( 1 \right) 
	&  \gamma_{\Delta_N y, \Delta y} \left( 2 \right)
	& \cdots 
	& \gamma_{\Delta_N y, \Delta_N y} \left( p+2 - (N-1)  \right) \\
\gamma_{\Delta y} \left( 0 \right) 
	&  \gamma_{\Delta y} \left( 1 \right)
	& \cdots 
	& \gamma_{\Delta y} \left( p+2 - N  \right)
 \end{array} 
	\right),
\end{align}
}
%
%
%
%
%
%
%
%
%
%
%
%
%
%
%
%
%
%
%
%
%
%
%
%
%
%
%
%
we see the following relations 
\begin{align}
\Gamma_{rp, c}^{(2 + h)} &= \Gamma_{rp, c}^{(2)}(h) \qquad \mbox{for} \quad h = 1, ..., m-2  \nonumber \\
\Gamma_{rp,c}(h) &= A_c^h \Gamma_{rp, c} \qquad \mbox{for} \quad  h = 1, 2, ...  \nonumber \\
\Gamma_{rp, c}(2 +h) &= A_c^h\Gamma_{rp, c}^{(2)}, \label{eq: retrieve Gamma_rpc cols}
\end{align}
where by $\Gamma_{rp, c}^{(j)}(h)$, we denote the $j$-th column block of $\Gamma_{rp, c}(h)$. The first equation follows from the structure of the autocovariances of the states, i.e. $\Gamma_{rp,c}(h)= \E x_{t+h} x_t'$ for $h \in \mathbb N_0$, the second equation follows from the Lyapunov equations. Hence, we receive all columns of $\Gamma_{rp, c}$ by using the recursions in (\ref{eq: retrieve Gamma_rpc cols}) and therefore of $\Gamma_{rp} = c^{-1} \Gamma_{rp, c} {c^{-1}}'$. Finally, again by using the Lyapunov equations we have all second moments of $(\Delta y_t )_{t \in \mathbb Z}$ and $(u_t^\mathcal{S})_{t \in \mathbb Z}$.\\
Now $\Sigma_{\nu}$ retained by using the ``high frequency Yule-Walker type equations'', 
that is,
{\scriptsize
	\begin{eqnarray}
	\label{eq:7.18a_FD}
	&& \Delta y_{t} - \alpha  \beta^{\prime} y_{t-1} 
	- \Phi_1 \Delta y_{t-1} 
	- \dots - \Phi_{p-1} \Delta y_{t-p+1}  = \nu_t \nonumber \\
	&& 
	\Delta y_{t} \Delta y_{t}^{\prime} - \alpha  \beta^{\prime} y_{t-1} \Delta y_{t}^{\prime}
	- \Phi_1 \Delta y_{t-1} \Delta y_{t}^{\prime}
	- \dots - \Phi_{p-1} \Delta y_{t-p+1} \Delta y_{t}^{\prime} = \nu_t \Delta y_{t}^{\prime} \nonumber \\
	&& \underbrace{ \mathbb{E}_{} 
	\Delta y_{t} \Delta y_{t}^{\prime} 
	}_{\gamma_{\Delta y} (0) } 
	- \alpha  \underbrace{ \mathbb{E}_{}  \beta^{\prime} y_{t-1} \Delta y_{t}^{\prime}
	}_{\gamma_{\beta y \Delta y} (1) }
	- \Phi_1  \underbrace{ \mathbb{E}_{} \Delta y_{t-1} \Delta y_{t}^{\prime} 
	}_{\gamma_{\Delta y} (1) }
	- \dots - \Phi_{p-1} 
	\underbrace{\mathbb{E}_{} \Delta y_{t-p+1} \Delta y_{t}^{\prime}
	}_{\gamma_{\Delta y} (p-1) }
	= \underbrace{ \mathbb{E}_{} \nu_t \Delta y_{t}^{\prime} }_{\Sigma_{\nu} } \  . 
	\end{eqnarray} }
Hence, also generic identifiability of $\Sigma_{\nu}$ is established.\\
Finally we prove continuity of $\pi^{-1}$. This involves two steps: 1. The continuity of the mapping from the observed second moments to the parameters of a canonical minimal realisation $(\bar A_{b, c}, \bar B_{b, c}, \bar C_{b, c})$ (say the echelon form):\\
Recall that the set of transfer functions with McMillan-degree $m$, call it $\tilde M (m)$, can be decomposed in disjoint pieces corresponding to different Kronecker indices summing up to $m$. The set of transfer functions where the first $m$ rows of the Hankel matrix are a basis of the row space of the Hankel matrix is generic in $\tilde M (m)$ (w.r.t. the pointwise topology for $\tilde M(m)$ \citep[see][p. 65]{HannanAndDeistler2012}). This set is also called the ``generic neighbourhood''. As has been shown in Step 5 above, $\Gamma_{r, pc}$ from equation (\ref{matrix:xx}) has full rank $m$. We know that the linear dependencies in the Hankel matrix of the transfer function, say $\tilde{\mathcal H}$, and the Hankel matrix of the second moments, say $\tilde{\mathcal H}_\gamma$, are the same \citep[for the definitions see][]{AndersonEtAl2016a}. Now since $\Gamma_{rp, c}$ is the upper left $m \times m$ block of $\tilde{\mathcal H}_\gamma$, we know that the first $m$ rows of $\tilde{\mathcal H}$ are a basis of the row space of $\tilde{\mathcal H}$. Therefore $\Theta_I$ is a subset of the generic neighbourhood.\\
2. Note that from a given minimal realisation $(\bar A_{b, c}, \bar B_{b, c}, \bar C_{b, c})$ of $\theta \in \Theta_I$ all transformations involved in the retrieval algorithm described above are continuous. 
\end{proof}
\begin{lemma}
	\label{lemma:CandO_FD}
	Suppose that Assumptions~\ref{ass: Theta} and
	\ref{a: g-ident} hold. 
The matrix
		\begin{eqnarray}
		\label{eq:def:matrixO}
\mathcal{O}
		&=&
			\left(
			\begin{array}{c}	
				S_{\beta }^{}
				{A}_{b,c}  
				\\ 
				S_{\Delta_N y}^{}
				{A}_{b,c} 
				\\ 
				S_{\Delta_N y}^{}
				{A}_{b,c}^2 \\ 
				\vdots \\
				S_{\Delta_N y}^{}
				{A}_{b,c}^{n(p-1)}
			\end{array}
			\right) 
		\ 	
\end{eqnarray} 
is of full column rank $m=r+n(p-1)$.
\end{lemma}
\begin{proof}
	The proof is very similar to the proof that the observability matrix is of full rank in \citet{AndersonEtAl2016a}[Proof of Theorem~7, page~823]. 	
	Since the matrix $c$ is of full rank $m$ we are allowed to consider $A^N$ and $A$. To see this, 
	let $\tilde{q}_i$ now denote an eigenvector of $c {A} c^{-1}$ with eigenvalue $\lambda_i$, then 
	$\left(c {A}^N c^{-1} \right) \tilde{q}_i  = c {A}^{N-1} c^{-1} c {A} c^{-1} \tilde{q}_i = \lambda_i c {A}^{N-1} c^{-1} \tilde{q}_i = \lambda_i^N \tilde{q}_i$. 
	In addition, if $q_i$ is an eigenvector of A, then $\tilde{q}_i= c q_i$ is an eigenvector of $A_{b,c}$.
	%
	Moreover, $A_{b,c}^{j} \tilde{q}_i = c A^{jN} c^{-1} c q_i = \lambda_i^{Nj} c q_i$. 
	The eigenvalues of ${A}$ are such that 
	$\lambda_i \not=\lambda_j$ implies $\lambda_i^2 \not=\lambda_j^2$, the eigenvectors of ${A}$ and ${A}^2$ coincide. 
	To see this, let $q_i \in \mathbb{R}^{m}$ and $\lambda_i \in \mathbb{R}^{} $ denote an eigenvector and an eigenvalue of the matrix ${A}$. Then, ${A} q_i  = \lambda_i q_i$ and 
	${A}^2 q_i  = {A} {A} q_i = \lambda_i {A} q_i = \lambda_i^2 q_i$; for $N>2$ this works in the same way. Therefore it is sufficient to look at the eigenvectors and eigenvalues of the matrix $A$.
	%
	%
	Similar to \citet{AndersonEtAl2016a}[Lemma~2] we have shown in the proof of the above Theorem~\ref{thm: Minimal} that the
	first $r+n$ components of an eigenvector of ${A}$ or $c A c^{-1}$ are not equal to a vector of zeros. Therefore, by the Popov-Belevitch-Hautus (PBH)-eigenvector test 
	\citep[see, e.g.,][page~135]{Kailath1980}, the matrix ${\mathcal{O}}$ has full {\em column} rank $r+n(p-1)$. That is,
\begin{eqnarray}
\label{eq:pbhtest:O}
 \left(
 \begin{array}{c}
 \left(
A^N - \lambda_i^N I_m 
 \right) \\
 \left(
 \begin{array}{c}
 S_{\beta} \\ 
 S_{\Delta y} 
 \end{array}
 \right) A^N \\
 \end{array}
 \right) q_i
 =
 \left(
 \begin{array}{c}
0_{m \times 1}
 \\
 \lambda_i^N
 \left[ q_i \right]_{1:n+r} \not= 0_{n+r \times 1}
 \end{array}
 \right) 
 \ .
\end{eqnarray}
\end{proof}	
%
%
\section{Proof of Theorem \ref{thm: Manfreds Question}}
\begin{proof}
The proof is constructed as follows: For each of the identifiability conditions in Assumption \ref{a: g-ident}, we suppose that \textbf{(Ij)} is violated for $j = 1, ..., 6$ and show that there exists no ``observationally equivalent'' $\theta \in \Theta_I$.\\
Suppose \textbf{(I1)} or \textbf{(I2)} are violated for $\theta_{\neg I}$, then it follows that the McMillan degree of $\tilde k (\tilde z)$ is less than $m$. Hence there exists no $\theta \in \Theta_I$ with the same auto-covariance function $\tilde \gamma$, which is granted by $\tilde{K} (0) = I_{\tilde{n}} $.\\
Suppose \textbf{(I3)} or \textbf{(I4)} are violated, then the minimal realisation of $\bar A_{b, c}$, which is directly obtained from $\tilde \gamma$ has eigenvalues $\lambda_i^N = \lambda_j^N$ for some $i \neq j$, and thus \textbf{(I4)} is violated.\\
Suppose that neither of the conditions in \textbf{(I5)} hold, then by equations (\ref{eq: tilde C_b times q}), we have $C_b q = 0$ and the system is not observable (and therefore of McMillan degree smaller than $m$).\\
Suppose that condition \textbf{(I6)} is not satisfied, then after going through steps Steps 1 and 2 of the retrieval algorithm in the proof of Theorem \ref{thm: identifiability in steps}, we obtain in equation (\ref{eq: S_nf A q undsoweiter}) that $S_{n_f}^{(1)} q_i = 0$ for some $i$ and therefore we are outside of $\Theta_I$ already. 
\end{proof}

\section{Proof of Theorem~\ref{thm:identify:mu0mu1}}
\label{appa:deterministic}
By considering the Granger-Representation-Theorem for the solution on $\mathbb Z$ in \citet{bauerwagner2012} [equation (26)], \citep[see also][Theorem~4.3]{Johansen1995}, we obtain 
\citep[see also][Exercise~4.5]{hansen1998workbook}
\begin{align}
    \E \Delta y_t &= 
    C(\mu_0 + \mu_1 t) + M_c \mu_1 \label{eq: E Delta y_t} \ , \\
    \qquad \mbox{with  } C &:= \beta_\bot \bigg(\alpha_\bot' \left( I_n - \sum_{j = 1}^{p-1} \Phi_j \right) \beta_\bot\bigg)^{-1} \alpha_\bot' \ , 
\end{align}
and $M_c$ is the limit of the stable part of the impulse responses in the particular solution defined in \citet{bauerwagner2012} [equation (8)] (this is  $k_\bullet(1)$ in \citet{bauerwagner2012} equation (26)).\\
Let $\bar \alpha' = (\alpha' \alpha)^{-1}\alpha'$, then 
\begin{align}
    \E \beta'y_{t} &=
    \bar \alpha' \bigg[C(\mu_0 + \mu_1 t) + M_c \mu_1 - \mu_0 - \mu_1 t \ , \nonumber \\ 
    &\hspace{3cm} - \sum_{j = 1}^{p-1} \Phi_j\big( C(\mu_0 + \mu_1 (t-j)) + M_c \mu_1\big) \bigg] \ . \label{eq: E beta y t-1} 
\end{align}
For $N=1$, both moments in equations (\ref{eq: E Delta y_t}) and (\ref{eq: E beta y t-1}) are observable (or consistently estimable from observed data) for the stock as well as for the flow case. In addition, for $N>1$ we get
\begin{align}
    \E \beta'y_{t-\ell} & =
    \bar \alpha' \bigg[C(\mu_0 + \mu_1 (t - \ell ) ) + M_c \mu_1 - \mu_0 - \mu_1 (t-\ell)  \nonumber \\ 
    &\hspace{1.cm} - \sum_{j = 1}^{p-1} \Phi_j\big( C(\mu_0 + \mu_1 (t-j)) + M_c \mu_1\big) \bigg] \ \text{for the stock case, and }  \label{eq: E beta y stock}    \\
    \E \beta'w_{t} &=
    \E \sum_{\ell=0}^{N-1} \beta'y_{t-\ell}  \ , \ \ \text{for the flow case} \ . \label{eq: E beta y flow}   
\end{align}
\begin{align}
    \E \Delta_N y_t &= \sum_{j=0}^{N-1} \left( C(\mu_0 + \mu_1 ({t-j})) + M_c \mu_1 \right) \label{eq: E Delta y_t N} 
    \nonumber
    \\ &= N C \mu_0 +  C \mu_1 \left( \sum_{j=0}^{N-1} ({t-j})   \right) + NC M_c \mu_1
    \ , \ \ \text{for the stock case, and } \\
\E \Delta_N^{\Sigma} y_t &= \sum_{\ell=0}^{N-1} 
    \E \Delta_N y_{t-\ell} = 
\sum_{\ell=0}^{N-1} \sum_{j=\ell }^{\ell+ N-1}
\left( C(\mu_0 + \mu_1 ({t-j})) + M_c \mu_1 \right) \nonumber \\
& = 
N^2 C \mu_0 + 
C \mu_1 
\left( \sum_{\ell=0}^{N-1} \sum_{j=\ell }^{\ell+ N-1} 
 ({t-j}) \right) + 
 N^2 M_c \mu_1 
\ , \ \ \text{for the flow case .} 
\end{align}
%
%
%
%
%
%
%
%
%
As already stated in Section~\ref{sect:deter}, we first remove the deterministic trends from 
 $\E \beta^{\prime} y_t$ and $\E \Delta_N y_t$
for the stock case or
$\E \beta^{\prime} w_t = 
\E \sum_{j=0}^{N-1} \beta^{\prime} y_{t-j}$
and $\E \Delta_N^{\Sigma} y_t = 
\E \sum_{j=0}^{N-1} \Delta_N y_{t-j}$ for the flow case. 
$\beta$ follows from mixed frequency data as demonstrated in \citet{chambers2020frequency}.
Then the parameters of the model without deterministic terms are generically identified as shown in the main text. For the five cases arsing from model
(\ref{eq:vecm:det1}) we get:

\textbf{Case $H_2(r)$:} ($\mu_0 = \mu_1 =0$) 
No deterministic terms. This case was shown above.\\
\textbf{Case $H_1(r)$:}  ($\mu_1 = 0$ and $\mu_0 \neq 0$.) We have a linear trend in $\mathbb{E} y_t$, and a constant in $\mathbb{E} \Delta y_t$, and a constant in $\mathbb{E} \beta' y_{t}$. 
Let the matrix $ S_C $ select $n-r$ basis rows of $
C$ (e.g., $S_C = \alpha_\bot^{\prime}$ can be used).
In this case it holds that 
\begin{align*}
 \E  \begin{pmatrix}
       \beta' y_t \\
      S_C \Delta_N y_t
    \end{pmatrix}  &= 
    \begin{pmatrix}
      \bar \alpha' ((I_n - \sum_{j = 1}^{p-1} \Phi_j ) C - I_n ) \\
      N S_C C , 
    \end{pmatrix} \mu_0 \ , 
 \ \ \text{for the stock case, and } \nonumber \\ 
  \E  \begin{pmatrix}
       \beta' w_t \\
      S_C \Delta_N^\Sigma y_t
    \end{pmatrix}  &= 
    \begin{pmatrix}
     N \bar \alpha' ((I_n - \sum_{j = 1}^{p-1} \Phi_j ) C - I_n ) \\
      N^2 S_C C 
    \end{pmatrix} \mu_0 \ , 
      \ \text{for the flow case} \ ,
\end{align*}
where the matrix on the LHS before $\mu_0$ is of rank $n$ since the first rowblock has rank $r$ and $C$ has rank $n - r$ and both are mutually orthogonal as $\alpha_\bot ' \alpha = 0$ and therefore (with $ \alpha_\bot'$ the last term of $C$)
\begin{align*}
    \alpha_\bot'\big((I_n - \sum_{j = 1}^{p-1} \Phi_j ) C - I_n \big)'\bar \alpha = 0 \ .
\end{align*}
From this we can compute $\mu_0$.\\
\textbf{Case $H_1^*(r)$:} ($\mu_1 = 0$, $\mu_0 \neq 0$ and $\alpha_\bot' \mu_0 = 0$) no linear trend but constant in $\E y_t$, constant in $\mathbb{E} \beta^{\prime} y_{t-1}$, and $\mathbb{E}\Delta y_t = 0$. This results in $r$ parameters in the cointegration equation. In this case we write $\mu_t = \alpha \rho_0$, where $\rho_0 \in \mathbb{R}^{r}$.
Note that 
$C \mu_t = C \alpha \rho_0 = 0$, which results in 
\begin{align*}
  \E  \begin{pmatrix}
       \beta' y_{t} \\
      \Delta_N y_t
    \end{pmatrix}  &= 
    \begin{pmatrix}
       \bar \alpha' ((I_n - \sum_{j = 1}^{p-1} \Phi_j ) C - I_n ) \\
      N C
    \end{pmatrix} 
    \alpha \rho_0
    \nonumber \\
    &= 
    \begin{pmatrix}
      - \bar \alpha' \alpha \rho_0 \\
      N 0
    \end{pmatrix} 
    = 
    \begin{pmatrix}
      - \rho_0 \\
      0
    \end{pmatrix} 
    \  ,   \  \  \text{for the stock case.}
\end{align*}
For the flow case we get
\begin{align*}
  \E  \begin{pmatrix}
       \beta' w_{t} \\
      \Delta_N^{\Sigma} y_t
    \end{pmatrix}  &= 
    \begin{pmatrix}
     N \bar \alpha' ((I_n - \sum_{j = 1}^{p-1} \Phi_j ) C - I_n ) \\
      N^2 C
    \end{pmatrix} 
    \alpha \rho_0
    \nonumber \\
    &= 
    \begin{pmatrix}
      - N \bar \alpha' \alpha \rho_0 \\
      N^2 0
    \end{pmatrix} 
    = 
    \begin{pmatrix}
      - N \rho_0 \\
      0
    \end{pmatrix} 
    \  .
\end{align*}

This allows to uniquely retrieve the unkown parameter $\rho_0$ from  
 $\E \beta' y_{t-1}$.\\
\textbf{Case $H(r)$:}  ($\mu_0 \in \mathbb{R}^n$ and  $\mu_1 \neq 0$) quadratic trend in $\mathbb{E} y_t$, linear trend in $\mathbb{E} \beta^{\prime} y_t$, linear trend in $\mathbb{E} \Delta y_t$. For this the stock case we get the system of linear equations
\begin{align*}
  \E  \begin{pmatrix}
       \beta' y_{t} \\
     S_C  \Delta_N^{\Sigma} y_t \\
      \beta' y_{t-N} \\
     S_C \Delta_N^{\Sigma} y_{t-N} \\
    \end{pmatrix}  &= M_\mu
     \left( 
     \begin{array}{c}
          \mu_0 \\
          \mu_1
     \end{array}
     \right) \ , 
\end{align*} 
%
where 
{\tiny
\begin{align*} 
M_{\mu}^{\mathcal{S}} & :=   \begin{pmatrix}
      \bar \alpha' ( C - I_n - \sum_{j = 1}^{p-1} C \Phi_j )  
      &|& 
\bar \alpha' \left[ C  t 
+ M_c - 
 t
I_n - \sum_{j = 1}^{p-1} \Phi_j \left( C 
 (t-j) 
- M_c \right) \right]
      \\
     N S_C C &|& S_C C \left(   
     \sum_{\ell =0 }^{N-1}
     (t-\ell)  + N M_c \right) \\ \hline
 \bar \alpha' ( C -  I_n - \sum_{j = 1}^{p-1} C \Phi_j )  
      &|& 
\bar \alpha' \left[ C 
(t-N) + M_c - 
(t-N) I_n - 
\sum_{j = 1}^{p-1} \Phi_j \left( C (t-N-j) - M_c \right) \right]
      \\
     N S_C C & | & S_C C
 \left(     
 \sum_{\ell=0}^{N-1} 
 (t-N-\ell)    
     + N M_c \right)
    \end{pmatrix} \\
    &= \begin{pmatrix}
M_{\mu 11 }^{\mathcal{S}} & M_{\mu 12}^{\mathcal{S}} \\
M_{\mu 21 }^{\mathcal{S}} & M_{\mu 22}^{\mathcal{S}}
     \end{pmatrix} \ .
\end{align*} 
}

For the flow case we have
\begin{align*}
  \E  \begin{pmatrix}
       \beta' w_{t} \\
     S_C  \Delta_N^{\Sigma} y_t \\
      \beta' w_{t-N} \\
     S_C \Delta_N^{\Sigma} y_{t-N} \\
    \end{pmatrix}  &= M_\mu
     \left( 
     \begin{array}{c}
          \mu_0 \\
          \mu_1
     \end{array}
     \right) \ , 
\end{align*} 
%
where 
{\scriptsize
\begin{align*} 
& M_{\mu}^{\mathcal{F}}  := \left( \begin{matrix}
      N \bar \alpha' ( C - I_n - \sum_{j = 1}^{p-1} C \Phi_j ) &|& 
      \\
     N^2 S_C C &|& 
     \\ \hline 
 N \bar \alpha' \left(  C -  I_n -  \sum_{j = 1}^{p-1} C \Phi_j\right) &|& 
\\
     N^2 S_C C &|&
\end{matrix} \right.  
\\
& \quad \left. \begin{matrix}
\bar \alpha' \left[ C 
\sum_{\ell=0}^{N-1} (t-\ell) 
+ N M_c - 
\left( \sum_{\ell=0}^{N-1} (t-\ell) \right) 
I_n - \sum_{j = 1}^{p-1} \Phi_j 
\left( C 
\sum_{\ell=0}^{N-1} (t-\ell) 
- N M_c \right) \right] 
      \\
 S_C \left( 
     \left( \sum_{\ell=0}^{N-1} 
     \sum_{j=0}^{N-1}
     (t-j-\ell) \right) C  + N^2 M_c \right) 
     \\ \hline
\bar \alpha' \left[ C 
\sum_{\ell=0}^{N-1}
(t-N-\ell) + N M_c - 
\sum_{\ell=0}^{N-1} (t-\ell-N) I_n - 
\sum_{j = 1}^{p-1} \Phi_j \left( C 
\left( \sum_{\ell=0}^{N-1} 
(t-N-j-\ell) \right) - N M_c \right) \right] 
\\
     S_C \left( 
  \left(    
 \sum_{\ell=0}^{N-1} 
 \sum_{j=0}^{N-1}
 (t-N-\ell-j) \right)    
     C + N^2 M_c \right) 
\end{matrix} \right) 
\\
&= \begin{pmatrix}
M_{\mu 11 }^{\mathcal{F}} & M_{\mu 12}^{\mathcal{F}} 
\\
M_{\mu 21 }^{\mathcal{F}} & M_{\mu 22}^{\mathcal{F}}
     \end{pmatrix} \ .
\end{align*} 
}

\noindent
The matrix $ S_C $ selects $n-r$ basis rows of $
C$ 
(e.g., $S_C = \alpha_\bot^{\prime}$ can be used)
and $M_{\mu 11 }^{\cdot}= M_{\mu 21 }^{\cdot}$. The matrix $M_{\mu 11 }^{\cdot}$ is of full column rank $n$ by the above calculations where we set $\mu_1=0$. The determinant of a blocked matrix is given by 
{\small
$$ \det M_{\mu}^{\cdot} = \det M_{\mu 11 }^{\cdot}  \det \left( M_{\mu 22}^{\cdot} -  
\underbrace{M_{\mu 21 }^{\cdot} M_{\mu 11 }^{{\cdot} -1}}_{I_n} M_{\mu 12}^{\cdot} \right)=
\det M_{\mu 11 }^{\cdot}  \det \left( M_{\mu 22}^{\cdot} - M_{\mu 12}^{\cdot} \right) \ .$$ } Note that $M_{\mu 22}^{\mathcal{S}} - M_{\mu 12}^{\mathcal{S}} = 
\begin{pmatrix}
      N \bar \alpha' ( C - I_n - \sum_{j = 1}^{p-1} C \Phi_j )  
      \\
     N^2 S_C C 
    \end{pmatrix}  
 = N M_{\mu 11}^{\mathcal{S}}$ for the stock case, while 
 for the flow case
$M_{\mu 22}^{\mathcal{F}} - M_{\mu 12}^{\mathcal{F}} = 
\begin{pmatrix}
      N^2 \bar \alpha' ( C - I_n - \sum_{j = 1}^{p-1} C \Phi_j )  
      \\
     N^3 S_C C 
    \end{pmatrix}  
 = N M_{\mu 11}^{\mathcal{F}}$.
Therefore, $\det M_{\mu}^{\cdot}  \not=0$. This allows to uniquely solve for $\mu_0$ and $\mu_1$.\\

%
%
%
%
%
%
%
%
%
%
%
%
%
\textbf{Case $H^*(r)$:} ($\mu_0 \in \mathbb{R}^n$, $\mu_1 \neq 0$, with $\alpha_\bot' \mu_1 =0 $) linear trend in $\mathbb{E} y_t$, linear trend in $\mathbb{E} \beta^{\prime} y_t$, and constant in $\mathbb{E} \Delta y_t$. $\mu_0 $ contains $n$ free parameters, while $\mu_t = \mu_0 + \alpha \rho_1 t$, where $\rho_1 \in \mathbb{R}^{r}$.
Hence, 
$C \mu_t = C \mu_0 +  C \alpha \rho_1 t = C \mu_0$. 
This results in 
\begin{align*}
  \E  \begin{pmatrix}
       \beta' y_{t} \\
     S_C \Delta_N y_t \\
      \beta' y_{t-N} \\
    \end{pmatrix}  &= 
    \left[ 
    M_\mu^{\mathcal{S}} \right]_{1:r+n,1:2(r+m)}
     \left( 
     \begin{array}{c}
          \mu_0 \\
          \alpha \rho_1
     \end{array}
     \right)  \nonumber \\
&= 
     \underbrace{ \left[ M_\mu^{\mathcal{S}} \right]_{1:r+n,1:2(r+n)}
     \left( 
     \begin{array}{cc}
         I_n & 0_{n \times r} \\
         0_{n \times n}   &  \alpha 
     \end{array}
     \right) }_{M_{\mu,H_1^*(r)}^{\mathcal{S}} }
     \left( 
     \begin{array}{c}
          \mu_0 \\
          \rho_1
     \end{array}
     \right)
     \ .     
\end{align*} 
%
Since $M_{\mu}^{\mathcal{F}} $ has full rank $2n$, the matrix 
$M_{\mu H_1^*(r)}^{\mathcal{F}}$ has rank $r+n$. This allows to uniquely obtain $\mu_0$ and $\rho_1$. For the flow case this works equivalently. 
\end{document}